\documentclass[12pt]{article}

\usepackage{amsmath}
\usepackage{amssymb}
\usepackage{amsthm}
\usepackage{hyperref}
\usepackage{enumerate}
\usepackage{graphicx}
\usepackage{color}
\usepackage{mathtools}
\usepackage{float}
\usepackage{tikz}

\newtheorem{thm}{Theorem}[section]
\newtheorem{co}[thm]{Corollary}

\newtheorem{lem}[thm]{Lemma}
\newtheorem{conj}[thm]{Conjecture}
\newtheorem{assumption}[thm]{Assumption}

\newtheorem{pr}[thm]{Proposition}

\newtheorem{definition}[thm]{Definition}
\newenvironment{de}{\begin{definition}\rm}{\end{definition}}
\newtheorem{example}[thm]{Example}

\newtheorem{remark}[thm]{Remark}
\newenvironment{rem}{\begin{remark}\rm}{\end{remark}}
\newtheorem{algorithm}[thm]{Algorithm}

\newcommand{\renyi}{R\'enyi\ }

\newcommand{\epl}{\varepsilon}

\title{R\' enyi Entropy Rate of Stationary Ergodic Processes~\footnote{A preliminary version~\cite{wu17} of this work has been presented in IEEE ISIT 2017.}~\thanks{This work is supported by the Research Grants Council of the Hong Kong Special Administrative Region, China, under Project 17301017 and Project 17304121, the National Natural Science Foundation of China, under Project 61871343 and 61902380, and the Beijing Nova Program, under Grant Z201100006820061.}}

\author{Chengyu Wu\textsuperscript{1}, Yonglong Li\textsuperscript{2}, Li Xu\textsuperscript{3}, Guangyue Han\textsuperscript{4} \\[2ex]
\textsuperscript{1}The University of Hong Kong, \textit{chengyuw@connect.hku.hk}\\
\textsuperscript{2} National University of Singapore, \textit{elelong@nus.edu.sg}\\
\textsuperscript{3} Chinese Academy of Sciences, \textit{lixu@ict.ac.cn }\\
\textsuperscript{4}The University of Hong Kong, \textit{ghan@hku.hk}\\}

\date{{\normalsize \today}}
\begin{document}\maketitle\thispagestyle{empty}

\begin{abstract}
In this paper, we examine the R\' enyi entropy rate of stationary ergodic processes. 
For a special class of stationary ergodic processes, we prove that the R\'enyi entropy rate always exists and can be polynomially approximated by its defining sequence; moreover, using the Markov approximation method, we show that the R\'enyi entropy rate can be exponentially approximated by that of the  Markov approximating sequence, as the Markov order goes to infinity. 
For the general case, by constructing a counterexample, we disprove the conjecture that the R\' enyi entropy rate of a general stationary ergodic process always converges to its Shannon entropy rate as $\alpha$ goes to $1$. 
\end{abstract}

\section{Introduction}
Let $\mathcal{Z}$ be a finite alphabet. Let $Z_1^n \triangleq (Z_1, Z_2, \ldots, Z_n)$ be a sequence of random variables over $\mathcal{Z}$ with distribution $\mu_n$ and let $z_1^n$ denote its realization. Given $\alpha \in \mathbb{R}$, the {\em $\alpha$-th order \renyi entropy} of $Z_1^n$, first suggested by Alfred \renyi \cite{re1961}, is defined as
\begin{displaymath}
H_\alpha(Z_1^n)=\left\{
\begin{array}{ll}
\vspace{0.2cm}
\displaystyle\hspace{-.7mm} \frac{\log \sum_{z_1^n} (\mu_n(z_1^n))^\alpha}{1-\alpha} &\textrm{if } \alpha \neq 1,\\
\displaystyle\hspace{-.7mm} H(Z_1^n) &\textrm{if } \alpha =1,
\end{array}
\right.
\end{displaymath}
where
$$
H(Z_1^n) \triangleq -\sum_{z_1^n} \mu_n(z_1^n) \log \mu_n(z_1^n)
$$
is the Shannon entropy of $Z_1^n$. An easy application of L'H\^{o}pital's rule shows that 
\begin{align} \label{cont_for_aver_renyi}
\lim_{\alpha\rightarrow 1} H_{\alpha}(Z_1^n)=H(Z_1^n).
\end{align}
\renyi entropy is a fundamental notion in a number of scientific and engineering disciplines, such as coding theory~\cite{cs1995}, chaotic dynamical systems~\cite{gr1983}, statistical mechanics~\cite{ki2008}, statistical inference~\cite{pa2005}, quantum mechanics~\cite{be2006}, multi-fractal analysis~\cite{ji2004}, economics~\cite{ha1975}, guessing~\cite{Arikan1998}, hypothesis testing~\cite{Bassat1978}, and so forth.

Now, consider a stationary stochastic process $Z=\{Z_n\}_{n=1}^\infty$ over the alphabet $\mathcal{Z}$. Let 
$$H(Z)\triangleq \lim_{n\rightarrow \infty} \frac{H(Z_1^n)}{n}$$ be the Shannon entropy rate of $Z$. 
Then, the {\em $\alpha$-th order \renyi entropy rate} $H_{\alpha}(Z)$ 
of $Z$ is defined as
$$
H_{\alpha}(Z) \triangleq \lim_{n \to \infty} \frac{H_{\alpha}(Z_1^n)}{n},
$$
when the limit exists. As opposed to R\'enyi entropy, which has been extensively studied, there has long been a lack of understanding on some basic properties of R\'enyi entropy rate. To name a few, first of all, 
the fundamental problem of the well-definedness of the R\'enyi entropy rate for a general stationary ergodic process remains unknown.
Second, regarding its connection with the Shannon entropy rate, given (\ref{cont_for_aver_renyi}), one is natually tempted to propose the following natural conjecture:
\begin{conj} \label{conj}
Let $Z$ be a stationary ergodic process. Then $$\lim_{\alpha\rightarrow 1} H_{\alpha}(Z)=H(Z).$$
\end{conj}
\noindent However, this conjecture is neither proved nor disproved in the literature. 

On the positive side, some special cases have been handled and feature clean solutions. When $Z$ is an independent and identically distributed (i.i.d.) process, $H_{\alpha}(Z)$ boils down to nothing but $H_{\alpha}(Z_1)$. 
For a finite-state ergodic Markov process $Z$, using the Perron-Frobenius theory (see, e.g., \cite{me2001,se2006}), it has been proved in \cite{ra2001} that
\begin{equation} \label{Markov-Case}
H_\alpha(Z)=\frac{\log \lambda_{\max}}{1-\alpha}
\end{equation}
and $H_\alpha(Z)$ converges to the Shannon entropy rate $H(Z)$ as $\alpha$ goes to $1$, 
where $\lambda_{\max}$ is the largest real eigenvalue of the $\lvert\mathcal{Z}\rvert\times \lvert\mathcal{Z}\rvert$-dimensional matrix $R=(r_{i,j})$ with
$
r_{i,j}=(P\{Z_{n+1}=j|Z_n=i\})^\alpha.
$
It turns out that similar results are also valid for mixing processes: for a weakly $\psi$-mixing process $Z$, it has been shown in \cite{NH2008} that $H_{\alpha}(Z)$ is well-defined for $\alpha\geq 1$ and $H_{\alpha}(Z)$ always goes to $H(Z)$ as $\alpha$ goes to $1$; on the other hand, using Kingman's subadditive ergodic theorem \cite{ki1968}, it has been proved in \cite{sz2001} that the \renyi entropy rate of any order exists for the so-called weakly mixing processes.

The contributions of this paper can be summarized as follows.
We first focus our attention on the \renyi entropy rate of a special family of stationary ergodic processes which contains hidden Markov processes~\cite{ep2002} as special cases. 
More precisely, we will examine a 
random process $Y$ under the ``uniform boundedness" and ``exponential forgetting" properties (see Section \ref{special_class} for details).
Using 	
a refined Bernstein blocking method~\cite{Bern1927},
we first show that the R\'{e}nyi entropy rate $H_\alpha(Y)$ exists,
and the convergence rate of
$H_{\alpha}(Y_1^n)/n$ to
$H_{\alpha}(Y)$ is $O(n^{-\gamma})$, where $0 < \gamma < 1$
can be
arbitrarily close to $1$. Note that for the special case when $\alpha=1$ (the Shannon entropy case), it is well known (see, e.g.,~\cite{ha2013}) that the convergence rate is $O(n^{-1})$. So, in some sense, the derived convergence rate is sharp. Borrowing results from the theory of nonnegative matrices, we
also establish that $H_{\alpha}(Y)$ can be exponentially approximated by the \renyi entropy rate of the approximating Markov process, as the Markov order goes to infinity. Undoubtedly, as opposed to the polynomial convergence rate of $H_{\alpha}(Y_1^n)/n$, this exponential convergence rate allows us to compute $H_{\alpha}(Y)$ more efficiently, at least for some special situations.

We then examine the \renyi entropy rate of general stationary ergodic processes, for which we show that Conjecture \ref{conj} is not true. 
Note that the answer to Conjecture \ref{conj} is clearly negative if the ergodicity assumption is dropped: the example in Section IV of \cite{ra2001} shows that for some reducible Markov chain $X$, $H_\alpha(X)$ fails to converge to $H(X)$ as $\alpha$ goes to $1$. Although the existing results for i.i.d., Markov~\cite{ma1999} and weakly $\psi$-mixing processes~\cite{NH2008} might suggest a positive answer to Conjecture \ref{conj}, we will 
construct a stationary ergodic counterexample whose R\' enyi entropy rate does not converge to the Shannon entropy rate as the R\' enyi order goes to $1$.
The main tool employed in the construction is the cutting and stacking method, which is a well-known method in ergodic theory but somehow attracts little attention in the field of information theory. 

The remainder of this paper is organized as follows. 
First, we focus our attention on the special random process $Y$ mentioned above.
We show in Section~\ref{convergence_of_average_renyi} that 
the normalized \renyi entropy $H_{\alpha}(Y_1^n)/n$ converges to $H_{\alpha}(Y)$ polynomially. By introducing the Markov approximation sequence, we prove in Section~\ref{Convergence_of_Markov}
that the \renyi entropy rate of this sequence of Markov chains does converge to $H_{\alpha}(Y)$, and moreover, the rate of convergence is exponential.
Next, we turn to the construction of the stationary ergodic counterexample that disproves Conjecture \ref{conj}.
Some preliminaries on the cutting and stacking method are given in Section~\ref{SCT_cutting_and_stacking_m}. Then, based on this method, the construction of our counterexample is presented in Section \ref{constru_of_counter}, followed by the derivation of several properties of the  counterexample in Section \ref{Properties_of_mu}. As elaborated on in Section~\ref{Non-existence_of_rate}, these properties immediately imply that as $\alpha$ goes to $1$, the \renyi entropy rate fails to converge to the Shannon entropy rate for the constructed stationary ergodic process.

\section{R\' enyi Entropy Rate of a Special Class of Random Processes} \label{special_class}


In this section, we focus on a stationary process $Y$ satisfying the following two conditions:
\begin{enumerate}
\item[$(i)$] {\em uniform boundedness}: there exist $C_L, C_U>0$ such that for any realization sequence $y_1^n$,
$$
C_L\leq p(y_n\vert y_1^{n-1})\leq C_U;
$$
\item[$(ii)$] {\em exponential forgetting}: for any fixed $\alpha$, there exist $C_F>0$ and $0 < \rho_F<1$ such that for any $k, \hat{k}\geq n$ and for any two realization sequences $y_1^k$ and $\hat{y}_1^{\hat{k}}$ with $y_1^n=\hat{y}_1^n$, it holds that
$$
\lvert p^{\alpha}(y_k\vert y_1^{k-1})-p^{\alpha}(\hat{y}_{\hat{k}}\vert \hat{y}_1^{\hat{k}-1})\rvert\leq C_F\rho_F^n;
$$
\end{enumerate}

A typical example satisfying the above conditions is given below.


\begin{example} \label{hidden_Mar}
A hidden Markov chain is a finite-state Markov chain observed through a discrete memoryless channel. To be more specific, let $\mathcal{X}$ be the input alphabet, $\mathcal{Z}$ be the output alphabet, $\{X_n\}_{n=1}^\infty$ be a finite-state Markov chain and $\{p(z\vert x): x\in \mathcal{X}, z\in \mathcal{Z}\}$ be the channel transition probabilities. Then the distribution of a hidden Markov process $Z$ is given by 
$$
p(z_1^n)=\sum_{x_1^n} p(x_1^n, z_1^n)=p(x_1)p(z_1 \vert x_1) \prod_{i=2}^n p(x_i\vert x_{i-1}) p(z_i\vert x_i)
$$
for any realization sequence $z_1^n$. 
If we further assume that $Z$ satisfies the following two conditions:
\begin{enumerate}
\item[(1)]  the input Markov chain is irreducible and aperiodic,
\item[(2)]  the channel transition probability matrix is strictly positive,
\end{enumerate}
then it has been verified in~\cite{an2006} that $\{Z_n\}_{n=1}^\infty$ satisfies Conditions (i) and (ii). Here, we remark that as special cases, i.i.d. processes and irreducible and aperiodic finite-state Markov chains also satisfy Conditions (i) and (ii).
\end{example}

In the remainder of this section, we will first prove that for any fixed $\alpha$, $H_{\alpha}(Y)$ exists
and the convergence rate of $\{H_{\alpha}(Y_1^n)/n\}_{n=1}^\infty$ is polynomial. 
Then, making use of the Markov approximation, we show that when $\rho_F$ is small enough, the \renyi entropy rate of the Markov approximating sequence converges exponentially to $H_{\alpha} (Y)$. Note that the requirement for $\rho_F$ to be small can be justified in some practical situations: for a binary symmetric channel operating at the high signal-to-noise ratio regime, or roughly, its crossover probability is ``close'' to $0$, it has been observed (see, e.g.,~\cite{an2006}) that $\rho_F$ is also ``close'' to $0$. 

Before moving to the next section, let us introduce the following definition.

\begin{de} \label{Mar_app}
For a stochastic process $X$, its $m$-th order {\em Markov approximation}~\cite{gray} is a stochastic process $X^{(m)}=\{X_n^{(m)}\}_{n=1}^\infty$ with distribution $p^{(m)}$
such that:
\begin{itemize}
\item $X^{(m)}$ is an $m$-th order Markov process, that is, for any realization $x_1^n$ with $n\geq m,$
$$
p^{(m)}(x_1^n)=p^{(m)}(x_1^m)\cdot p^{(m)}(x_{m+1}\vert x_1^m) \cdots p^{(m)}(x_n\vert x_{n-m}^{n-1});
$$
\item the $(m+1)$-dimensional distribution of $X^{(m)}$ and $X$ are the same, namely,
$$
p^{(m)}(x_1^{m+1})=p(x_1^{m+1}).
$$
\end{itemize}
\end{de}
\begin{rem} \label{Mar_same_cons}
If $X$ satisfies Conditions $(i)$ and $(ii)$, then for any $m$, $X^{(m)}$ also satisfies these two conditions with the same constants $C_L, C_U, C_F, \rho_F$ (which are independent of $m$).
\end{rem}
Throughout the remainder of this section, we will always assume that $\alpha \not = 1$ since $\alpha=1$ corresponds to the Shannon entropy rate case. Furthermore, we always use $Y$ to denote a stationary process satisfying Conditions $(i)$ and $(ii)$ and $Y^{(m)}$ to denote the $m$-th order Markov approximation of $Y$.

\subsection{Convergence of $\{H_{\alpha} (Y_1^n)/n\}$} \label{convergence_of_average_renyi}

The following theorem establishes the existence of the R\'{e}nyi entropy rate $Y$; moreover, it establishes the convergence of $H_{\alpha}(Y_1^n)/n$ to $H_{\alpha}(Y)$ and gives a rate of convergence. Here, we note from Remark \ref{Mar_same_cons} that the theorem also applies to the $m$-th order Markov approximation $Y^{(m)}$ for any $m\geq 1$.

\begin{thm} \label{poly-rate}
For any $0<\gamma<1$, there exists a constant $C$ such that for all $n$,
\begin{align*}
\left\lvert\frac{H_{\alpha}(Y_1^n)}{n}-H_{\alpha}(Y)\right\rvert\leq C n^{-\gamma}\hspace{-1.0mm}. \vspace{-1.0mm}
\end{align*}
\end{thm}
\begin{proof}
We only prove the theorem for the case $0\leq \alpha<1$, since the cases $\alpha<0$ and $\alpha>1$ can be similarly handled.
	
	For any constant $\gamma\in (0,1)$, let
	$$
	\lambda=\displaystyle\frac{1+\gamma}{2}, \; \beta=\frac{1-\gamma}{2}, \; p=n^\lambda, \; q=\frac{n^\beta}{2}, \; \omega=n^{1-\lambda}.
	$$
Now we use the Bernstein blocking method (see \cite{Bern1927}) to consecutively partition the sequence $y_1^n$ into small pieces of length $q$, $q$ and $p-2q$. To be more specific, define
\begin{align*}
\xi_i\triangleq p^\alpha\left( y_{(i-1)p+1}^{(i-1)p+q}\big\vert y_1^{(i-1)p} \right),\quad \eta_i\triangleq p^\alpha\left( y_{(i-1)p+q+1}^{(i-1)p+2q}\big\vert y_1^{(i-1)p+q} \right), \quad
\zeta_i\triangleq p^\alpha\left( y_{(i-1)p+2q+1}^{ip}\big\vert y_1^{(i-1)p+2q} \right)
\end{align*}
and their truncated versions
\begin{align*}
\hat{\xi}_i\triangleq p^\alpha\left( y_{(i-1)p+1}^{(i-1)p+q}\big\vert y_{(i-2)p+q+1}^{(i-1)p} \right),\quad
\hat{\eta}_i\triangleq p^\alpha\left( y_{(i-1)p+q+1}^{(i-1)p+2q} \right), \quad
\hat{\zeta}_i\triangleq p^\alpha\left( y_{(i-1)p+2q+1}^{ip}\big\vert y_{(i-1)p+q+1}^{(i-1)p+2q} \right).
\end{align*}
Then, using the fact that for $k \neq j$, the $y$-sequences associated with $\hat{\eta}_j\hat{\zeta}_j\hat{\xi}_{j+1}$ and $\hat{\eta}_k\hat{\zeta}_k\hat{\xi}_{k+1}$ are both of length $p=n^\lambda$ and their index sets are non-overlapping, we have
\begin{align}
\sum_{y_1^n}p^\alpha(y_1^n)&=\sum_{y_1^n}\xi_1\eta_1\zeta_1\xi_2\eta_2\zeta_2 \cdots \xi_\omega \eta_\omega \zeta_\omega \notag \\
&\overset{(a)}{\leq} \sum_{y_1^n} C_U^{\alpha q\omega} \eta_1 \zeta_1 \eta_2 \zeta_2 \cdots \eta_\omega \zeta_\omega \notag \\
&\overset{(b)}{\leq} \sum_{y_1^n} C_U^{\alpha q\omega} \left( \displaystyle\frac{C_U}{C_L} \right)^{\alpha q \omega} \hat{\eta}_1 \zeta_1 \hat{\eta}_2 \zeta_2 \cdots \hat{\eta}_\omega \zeta_\omega \notag \\
&\overset{(c)}{\leq} \sum_{y_1^n} \left( \displaystyle\frac{C_U^2}{C_L} \right)^{\alpha q \omega} \left( 1+\displaystyle\frac{C_F \rho_F^q}{C_L^\alpha} \right)^{(p-2q)\omega} \notag \cdot \hat{\eta}_1 \hat{\zeta}_1 \hat{\eta}_2 \hat{\zeta}_2 \cdots \hat{\eta}_\omega \hat{\zeta}_\omega  \notag \\
& \overset{(d)}{\leq} \left( \displaystyle\frac{C_U^2}{C_L} \right)^{\alpha q \omega} \left( 1+\displaystyle\frac{C_F \rho_F^q}{C_L^\alpha} \right)^{(p-2q)\omega}  \cdot\sum_{y_1^n} \left[ \left( \displaystyle\frac{1}{C_L^{\alpha}} \right)^{q \omega}
   (\hat{\eta}_1\hat{\zeta}_1\hat{\xi}_2)\cdots (\hat{\eta}_\omega\hat{\zeta}_\omega\hat{\xi}_{\omega+1}) \right] \notag \\
&= \left( \displaystyle\frac{C_U}{C_L} \right)^{2\alpha q \omega} \left( 1+\displaystyle\frac{C_F \rho_F^q}{C_L^\alpha} \right)^{(p-2q)\omega}  \cdot \left( \sum\nolimits_{y_1^{n^\lambda}} p^\alpha (y_1^{n^\lambda}) \right)^\omega, \label{bound-1}
\end{align}
where for $(a)$ and $(b)$, we have used Condition $(i)$ to drop all $\xi_i$'s and replaced all $\eta_i$'s by their truncated versions; for $(c)$, we have applied Conditions $(i)$ and $(ii)$ to replace all $\zeta_i$'s by their truncated versions; and for $(d)$, we have applied Condition $(i)$ to add $\hat{\xi}_2, \cdots ,\hat{\xi}_{\omega+1}$.

Taking logarithm and dividing both sides of (\ref{bound-1}) by $n$, we obtain
\begin{align}
\displaystyle\frac{\log\sum_{y_1^n} p^\alpha(y_1^n)}{n}&=\frac{\omega \log\sum_{y_1^{n^\lambda}}p^\alpha (y_1^{n^\lambda})}{n}+\displaystyle\frac{2\alpha q \omega \log\left(\displaystyle\frac{C_U}{C_L}\right)}{n} +\frac{(p-2q)\omega \log\displaystyle\left(1+\frac{C_F \rho_F^q}{C_L^\alpha}\right)}{n} \notag \\
& \leq \displaystyle\frac{\log \sum_{y_1^{n^\lambda}} p^\alpha (y_1^{n^{\lambda}})}{n^\lambda }
 +\alpha \log\displaystyle\left( \frac{C_U}{C_L} \right)n^{-\gamma}+\frac{C_F\rho_F^q}{C_L^\alpha}. \notag
\end{align}
Note that $0< \rho_F<1$ and $q=n^\beta/{2}$ implies
$$
\frac{C_F\rho_F^q}{C_L^\alpha}\leq \alpha \log\left(\frac{C_U}{C_L} \right)n^{-\gamma}
$$
for sufficient large $n$.
It then follows that
$$
\displaystyle\frac{\log\sum_{y_1^n} p^\alpha(y_1^n)}{n}\leq \displaystyle\frac{\log \sum_{y_1^{n^\lambda}} p^\alpha (y_1^{n^{\lambda}})}{n^\lambda }+2\alpha \log\left( \frac{C_U}{C_L} \right)n^{-\gamma},
$$
which immediately implies 
\begin{align} \label{upperbd_averaged_renyi}
\frac{H_{\alpha}(Y_1^n)}{n}\leq \frac{H_{\alpha}(Y_1^{n^\lambda})}{n^\lambda}+ {C}_1 n^{-\gamma}
\end{align}
for some constant $C_1$. Applying a parallel argument to the other direction, we obtain that for sufficiently large $n$,
\begin{align} \label{lowerbd_averaged_renyi}
\frac{H_{\alpha} (Y)}{n} \geq \frac{H_{\alpha} (Y_1^{n^\lambda})}{n^\lambda} + C_2 n^{-\gamma}
\end{align}
for some constant $C_2$.
Choosing $\widetilde{C}\triangleq \max \{C_1, C_2\}$, we derive from (\ref{upperbd_averaged_renyi}) and (\ref{lowerbd_averaged_renyi}) that
\begin{align} \label{abs_averaged_renyi}
\left\lvert \frac{H_{\alpha}(Y_1^n)}{n}- \frac{H_{\alpha}(Y_1^{n^\lambda})}{n^\lambda}\right\rvert<\widetilde{C}n^{-\gamma}.
\end{align}
Now consider any $m,n$ with $m>n\ge N$, where $N$ is a sufficiently large number to be determined later. Pick a number $\xi$ between $1$ and $\sqrt{2}$ (\emph{e.g.}, $5/4$). Let $t$ be the positive integer such that
$$t'\triangleq t+\log_\xi\log_\xi n-\log_\xi\log_\xi m\in [1,2).$$
Then, $\xi^{t'}\in [\xi,\xi^2)\subset (1,2)$ and $m^{\xi^{t'}}=n^{\xi^t}$.
Let
$$
\lambda_1=\xi^{-t'}, \; \gamma_1=2\xi^{-t'}-1, \; \lambda_2=\xi^{-1}, \; \gamma_2=2\xi^{-1}-1.
$$
Then, $0<\lambda_1,\gamma_1,\lambda_2,\gamma_2<1$ and 
\begin{align}\hspace{-1.0mm} \label{H_alpha_n_rate}
\left\lvert\frac{H_{\alpha}(Y_1^m)}{m}-\frac{H_{\alpha}(Y_1^n)}{n}\right\rvert &\leq  \biggl\lvert\frac{H_{\alpha}(Y_1^m)}{m}-\frac{H_{\alpha}(Y_1^{m^{\xi ^{t'}}})}{m^{\xi ^{t'}}}\biggr\rvert+\bigg\lvert \frac{H_{\alpha}(Y_1^n)}{n}-\frac{H_{\alpha}(Y_1^{n^{\xi ^t}})}{n^{\xi ^t}}\bigg\rvert \notag\\
&\le   \bigg\lvert\frac{H_{\alpha}(Y_1^m)}{m}-\frac{H_{\alpha}(Y_1^{m^{\xi ^{t'}}})}{m^{\xi ^{t'}}}\bigg\rvert+\bigg\lvert\frac{H_{\alpha}(Y_1^n)}{n}-\frac{H_{\alpha }(Y_1^{n^\xi})}{n}\bigg\rvert\hspace{0.5mm} \notag \\
&\textnormal{ \ \ \  } +\bigg\lvert \frac{H_{\alpha}\hspace{-0.1mm}(Y_1^{n^{\xi}})}{n^\xi}\hspace{-0.5mm}-\hspace{-0.5mm}\frac{H_{\alpha}\hspace{-0.1mm}(Y_1^{n^{\xi ^2}})}{n^{\xi ^2}}\bigg\rvert\hspace{-0.5mm}+\hspace{-0.3mm}\cdots\hspace{-0.3mm}+\hspace{-0.5mm}\bigg\lvert \frac{H_{\alpha}\hspace{-0.1mm}(Y_1^{n^{\xi ^{t-1}}})}{n^{\xi ^{t-1}}}\hspace{-0.5mm}-\hspace{-0.5mm}\frac{H_{\alpha}\hspace{-0.1mm}(Y_1^{n^{\xi ^t}})}{n^{\xi ^t}}\bigg\rvert \notag \\
&\overset{(e)}{\leq} \widetilde{C} m^{-\gamma_1/\lambda_1}+\widetilde{C} n^{-\gamma_2/\lambda_2}+\widetilde{C} n^{-\gamma_2/\lambda_2^2}+\cdots+\widetilde{C} n^{-\gamma_2/\lambda_2^t} \notag\\
&\overset{(f)}{\leq} \widetilde{C} m^{\xi^2-2}+\frac{\widetilde{C} n^{-(2-\xi)}}{1-n^{-(2-\xi)(\xi -1)}},
\end{align}
where $(e)$ follows from the inequality (\ref{abs_averaged_renyi}) and $(f)$ follows from the fact that $a^n-1 \geq n(a-1)$ for any $1<a<2$. For any given $\varepsilon > 0$, by choosing a sufficiently large $N$ such that
$$
\widetilde{C} N^{\xi^2-2}<\varepsilon/2 \quad \mbox{and} \quad \frac{\widetilde{C} N^{-(2-\xi)}}{1-N^{-(2-\xi)(\xi -1)}}<\varepsilon/2,
$$
we derive from (\ref{H_alpha_n_rate}) that
$$
\left\lvert \frac{H_{\alpha}(Y)}{m}-\frac{H_{\alpha}(Y)}{n}\right\rvert<\varepsilon
$$
for any $m>n\geq N.$
Thus the sequence $\{H_{\alpha,n}(Y)\}_{n\in \mathbb{N}}$ is Cauchy, and thereby convergent.
Furthermore, for any positive integers $k$ and $n$ with $n^{\gamma(\gamma-1)/1+\gamma}\leq \frac{1}{2}$, we have
\begin{align*}
 \bigg\lvert \frac{H_{\alpha}(Y_1^{n})}{n}-\frac{H_{\alpha}(Y_1^{n^{1/\lambda^k}})}{n^{1/\lambda^k}}\bigg\rvert & \le  \bigg\lvert \frac{H_{\alpha}(Y_1^n)}{n}-\frac{H_{\alpha}(Y_1^{n^{1/\lambda}})}{n^{1/\lambda}}\bigg\rvert
+\bigg\lvert \frac{H_{\alpha}(Y_1^{n^{1/\lambda}})}{n^{1/\lambda}}-\frac{H_{\alpha}(Y_1^{n^{1/\lambda^2}})}{n^{1/\lambda^2}}\bigg\rvert\\
& \textnormal{ \ \ \ }+\cdots+\bigg\lvert\frac{H_{\alpha}(Y_1^{n^{1/\lambda^{k-1}}})}{n^{1/\lambda^{k-1}}}-\frac{H_{\alpha}(Y_1^{n^{1/\lambda^k}})}{n^{1/\lambda^k}}\bigg\rvert\\
&\le  \widetilde{C} n^{-\gamma/\lambda}+\widetilde{C} n^{-\gamma/\lambda^2}+\cdots+ \widetilde{C} n^{-\gamma /\lambda^{k}}\\
&\le  \frac{\widetilde{C} n^{-\gamma}}{1-n^{\gamma-\gamma/\lambda}} \leq  \frac{\widetilde{C}}{n^{\gamma}-n^{\frac{2\gamma^2}{1+\gamma}}}\le \frac{2\widetilde{C}}{n^\gamma}.
\end{align*}
Then, letting $k$ tend to infinity, we have, for all sufficiently large $n$,
\begin{align*}
\left\lvert \frac{H_{\alpha}(Y_1^n)}{n}-H_{\alpha}(Y)\right\rvert\le 2\widetilde{C} n^{-\gamma}.
\end{align*}
The proof is then complete with an appropriately chosen common constant $C$ for all $n$.
\end{proof}


\subsection{Convergence of $\{H_{\alpha} (Y^{(m)})\}$} \label{Convergence_of_Markov}

When it comes to the computation of $H_{\alpha}(Y)$, the convergence of $\{H_{\alpha}(Y_1^n)/n\}$ as in Theorem~\ref{poly-rate} may be too slow to be applied in practice. 
In this section, we show that under some additional assumptions, $H_\alpha(Y)$ can be approximated by another exponentially convergent sequence that can be efficiently computed.

Our motivation comes from the fact that the R\'enyi entropy rate of a Markov process features a simple formula as in (\ref{Markov-Case}).
For any $m$, let $Y^{(m)}$ be the $m$-th order Markov approximation of $Y$. It is obvious form Definition \ref{Mar_app} that as $m$ goes to infinity, $Y^{(m)}$ converges in distribution to the original process $Y$; moreover, we note from \cite{ra2001} that $H_{\alpha} (Y^{(m)})$ is well-defined for all $m$. 
Indeed, we have the following theorem.

\begin{thm} \label{convergent}
$\lim\limits_{m\rightarrow \infty} H_\alpha(Y^{(m)})=H_\alpha(Y).$
\end{thm}

\begin{proof}
Note that for any $m$ and $n$, we have
{\small
\begin{align} \label{three-terms}
\lvert H_\alpha(Y^{(m)})-H_\alpha(Y)\rvert& \leq \bigg\lvert H_\alpha(Y)-\frac{H_{\alpha}(Y_1^n)}{n}\bigg\rvert  +\bigg\lvert\frac{H_{\alpha}(Y_{\textnormal{ \ \ \ \ }1}^{(m)n})}{n}-\frac{H_{\alpha}(Y_1^n)}{n}\bigg\rvert +\bigg\lvert\frac{H_{\alpha}(Y_{\textnormal{ \ \ \ \ }1}^{(m)n})}{n}-H_\alpha(Y^{(m)})\bigg\rvert.
\end{align}
}
\!\!We first deal with the first and third terms of the RHS of (\ref{three-terms}). It follows from Theorem~\ref{poly-rate} (applied to $Y$ and $Y^{(m)}$, which satisfy Conditions $(i)$ and $(ii)$) that for any given $\varepsilon > 0$, there exists $N_1 > 0$ such
that for any $n \geq N_1$ and any $m$,
$$
\left\lvert H_\alpha(Y)-\frac{H_{\alpha}(Y_1^n)}{n}\right\rvert \leq \varepsilon/3, \quad  \left\lvert H_\alpha(Y^{(m)})-\frac{H_{\alpha}(Y_{\textnormal{ \ \ \ \ }1}^{(m)n})}{n}\right\rvert \leq \varepsilon/3.
$$
Now, for the second term in the RHS of (\ref{three-terms}), we have
\begin{align}
&\left\lvert \frac{H_{\alpha}(Y_{\textnormal{ \ \ \ \ }1}^{(m)n})}{n}-\frac{H_{\alpha}(Y)}{n}\right\rvert \notag \\
&=\left\lvert{\displaystyle\frac{1}{(1-\alpha)n}\left[ \log\sum\nolimits_{y_{\textnormal{ \ \ \ \ }1}^{(m)n}} p^\alpha(y_{\textnormal{ \ \ \ \ }1}^{(m)n}) - \log\sum\nolimits_{y_1^n} p^\alpha(y_1^n) \right]} \right\rvert \notag \\
&= \left\lvert \frac{1}{(1-\alpha)n} \log \frac{\sum_{y_{\textnormal{ \ \ \ \ }1}^{(m)n}} p^\alpha (y_{\textnormal{ \ \ \ \ }1}^{(m)n})}{\sum_{y_1^n} p^\alpha(y_1^n)} \right\rvert \notag \\
&=\left\lvert \frac{1}{(1-\alpha)n} \log \frac{\sum_{y_1^n}p^\alpha(y_1^m) p^\alpha(y_{m+1}\vert y_1^m)\cdots p^\alpha(y_n\vert y_{n-m}^{n-1})}{\sum_{y_1^n}p^\alpha (y_1^m) p^\alpha (y_{m+1}\vert y_1^m) \cdots p^\alpha(y_n\vert y_1^{n-1})} \right\rvert. \notag
\end{align}
Replacing $p^{\alpha}(\cdot)$'s with simpler notations $a_i$'s and $b_i$'s, we continue to derive
\begin{align}
&\left\lvert \frac{H_{\alpha}(Y_{\textnormal{ \ \ \ \ }1}^{(m)n})}{n}-\frac{H_{\alpha}(Y_1^n)}{n}\right\rvert  \notag\\
 &=\left\lvert \frac{1}{(1-\alpha)n} \log \frac{\sum_{y_1^n} \prod_{i=1}^{n-m+1} a_i}{\sum_{y_1^n} \prod_{i=1}^{n-m+1} b_i} \right\rvert \notag \\
&= \left\lvert \frac{1}{(1-\alpha)n} \log \left[ 1+\frac{\sum_{y_1^n}\left(\prod_{i=1}^{n-m+1} a_i- \prod_{i=1}^{n-m+1} b_i\right)}{\sum_{y_1^n} \prod_{i=1}^{n-m+1} b_i} \right] \right\rvert \notag
\end{align}
\begin{align}
&\leq \left\lvert \frac{1}{(1-\alpha)n} \frac{\sum_{y_1^n}\left(\prod_{i=1}^{n-m+1} a_i- \prod_{i=1}^{n-m+1} b_i\right)}{\sum_{y_1^n} \prod_{i=1}^{n-m+1} b_i} \right\rvert \notag \\
&\leq \frac{1}{|1-\alpha|n} \frac{\sum_{y_1^n}\sum_{i=1}^{n-m+1}|a_i-b_i|}{\sum_{y_1^n} b_1 \cdots b_{n-m+1}} \notag \\
&\overset{(g)}{\leq} \frac{1}{\lvert 1-\alpha\rvert n} \frac{\lvert\mathcal{Y}\rvert^n(n-m+1)C_F\rho_F^m}{\lvert\mathcal{Y}\rvert^n C_L^{n-m+1}}, \notag\\
&\leq \frac{1}{\lvert 1-\alpha\rvert} \frac{C_F\rho_F^m}{C_L^{n-m+1}} \label{n-2m},
\end{align}
where Condition $(ii)$ is used in $(g)$.
Noting that $0<\rho_F<1$, $0<C_L<1$, we deduce that there exists $0<T<1$ such that $\rho_F<(C_L)^T$. Setting $n=\lfloor(T+1)m\rfloor-1,$ we have $n-m+1\leq Tm$, which, together with (\ref{n-2m}), implies that for the $\varepsilon$ given above, there exists an $N_2 > 0$ such that for all $m \geq N_2$,
$$
\left\lvert \frac{H_{\alpha}(Y_{\textnormal{ \ \ \ \ }1}^{(m)n})}{n}-\frac{H_{\alpha}(Y_1^n)}{n}\right\rvert \leq \varepsilon/3.
$$
It then follows from (\ref{three-terms})
that
\begin{equation} \label{diff-1}
\lvert H_\alpha(Y^{(m)})-H_\alpha(Y)\rvert \leq \varepsilon
\end{equation}
as long as $m \geq \max \{N_1+1, N_2\}$.
The desired convergence then follows from the arbitrariness of $\varepsilon$.
\end{proof}
\medskip


Having established the convergence of $\{H_{\alpha}(Y^{(m)})\}$ to $H_{\alpha}(Y)$, we now turn to its convergence rate.

First of all, for any fixed $m$, by a usual $m$-step blocking argument, we can transform $Y^{(m)}$ into a first-order Markov chain over a larger alphabet. To be more specific, define a new process $W^{(m)}=\{W_n^{(m)}\}_{n=1}^\infty$ such that
$$	
W_i^{(m)} = (Y_i^{(m)}, Y_{i+1}^{(m)}, \cdots, Y_{i+m-1}^{(m)}), \quad i=1, 2, \cdots.
$$
Apparently, $W^{(m)}$ is a first-order Markov chain over the alphabet $\mathcal{Y}^m$. Let $P^{(m)}=\{p^{(m)}_{ij}\}$ be the transition probability matrix of $W^{(m)}$, $R^{(m)}=\{(p_{ij}^{(m)})^\alpha\}$ be the matrix obtained by taking the $\alpha$-th power of each entry of $P^{(m)}$, and let $\lambda^{(m)}$ be the largest eigenvalue of $R^{(m)}$. Recalling from (\ref{Markov-Case}) that
$$
H_{\alpha}(Y^{(m)}) = \frac{\log \lambda^{(m)}}{1-\alpha},
$$
in order to derive the convergence rate of $H_{\alpha}(Y^{(m)})$, we only need to compare $\lambda^{(m)}$ and $\lambda^{(m+1)}$. Observing that $\lambda^{(m)}$ and $\lambda^{(m+1)}$ are the largest eigenvalues of two matrices whose dimensions are different, we first ``upscale" the matrix $R^{(m)}$ by viewing $Y^{(m)}$ as an $(m+1)$-th order Markov chain with the corresponding $\vert\mathcal{Y}\vert^{m+1} \times \vert\mathcal{Y}\vert^{m+1}$-dimensional transition probability matrix $\widetilde{R}^{(m)}$. It can then be readily verified that $\widetilde{R}^{(m)}$ has the same largest eigenvalue as $R^{(m)}$. Hence, it suffices for us to compare $R^{(m+1)}$ and $\widetilde{R}^{(m)}$, both of which are of dimension $\vert\mathcal{Y}\vert^{m+1}\times \vert\mathcal{Y}\vert^{m+1}.$

Assuming $\rho_F$ is small enough, the following theorem uses the previous observation to establish the exponential convergence of $\{H_\alpha(Y^{(m)})\}$ as $m\rightarrow \infty$.
	
\begin{thm} \label{expo-rate}
If $\rho_F < (C_L/C_U)^{2\alpha}$, 
then $H_\alpha(Y^{(m)}) \rightarrow H_{\alpha} (Y)$ exponentially as $m\rightarrow \infty$.
\end{thm}

\begin{proof}
According to Theorem \ref{convergent}, it suffices for us to show the exponential convergence of the sequence $\{H_{\alpha}(Y^{(m)})\}$. 

Let $\Delta_m=R^{(m+1)}-\widetilde{R}^{(m)}$. It follows from Condition $(ii)$ that the absolute value of each nonzero entry of $\Delta_m$ is upper bounded by $C_F\rho_F^m$. Applying the Collatz-Wielandt formula (see, e.g.,~\cite{horn1985}), we have
\begin{align} \label{eigenvalue_Rm+1}
\lambda^{(m+1)}&=\max_{x> 0} \min_i \displaystyle\frac{[R^{(m+1)}x]_i}{x_i} \notag\\
&=\max_{x> 0} \min_i \displaystyle\frac{[(\widetilde{R}^{(m)}+\Delta_m)x]_i}{x_i} \notag \\
&=\max_{x> 0} \min_i \left\{  { \displaystyle\frac{[\widetilde{R}^{(m)} x]_i}{x_i}+\displaystyle\frac{[\Delta_m x]_i}{x_i} } \right\}  \notag  \\
&\geq \max_{x>0}\left\{ { \min_i\displaystyle\frac{[\widetilde{R}^{(m)} x]_i}{x_i}+\min_j \displaystyle\frac{[\Delta_m x]_j}{x_j} } \right\},
\end{align}
where $x$ is a $\lvert \mathcal{Y}\rvert^{m+1}\times  1$ column vector and $[x]_j$ denote the $j$-th component of $x$. Let the vector $v$ be the right eigenvector of $\widetilde{R}^{(m)}$ such that the equality $\lambda^{(m)}= \min\limits_i \displaystyle\frac{[\widetilde{R}^{(m)} v]_i}{v_i}$ is achieved (Note from the Perron-Frobenius theorem \cite{me2001,se2006} that $v$ is a positive vector since $\widetilde{R}^{(m)}$ is a nonnegative irreducible matrix). Then we continue from (\ref{eigenvalue_Rm+1}) as follows:
\begin{align} \label{ineq-1}
\lambda^{(m+1)} & \geq \min_i\displaystyle\frac{[\widetilde{R}^{(m)} v]_i}{v_i}+\min_j \displaystyle\frac{[\Delta_m v]_j}{v_j} \notag \\
&=\lambda^{(m)}+ \min_j \displaystyle\frac{[\Delta_m v]_j}{v_j} \notag \\
&\geq \lambda^{(m)}- C_F \rho_F^m\cdot \lvert\mathcal{Y}\rvert \cdot \max_{i,j}\frac{v_i}{v_j},
\end{align}
where we have used the fact that each row of $\Delta_m$ has exactly $\lvert\mathcal{Y}\rvert$ strictly positive entries.

We now claim that for any $1\leq i,j\leq \lvert\mathcal{Y}\rvert^{m+1}$, $\displaystyle v_i/v_j$ can be bounded by
\begin{equation} \label{ratio-bounds-1}
\left( \frac{C_L}{C_U}  \right)^{2\alpha(m+1)}\leq \frac{v_i}{v_j}\leq \left( \frac{C_U}{C_L}  \right)^{2\alpha(m+1)}.
\end{equation}
To see this, first note that each entry $a_{i,j}$ of $(\widetilde{R}^{(m)})^{m+1}$ is of the form
$$
p^\alpha(y_{m+2}\vert y_1\cdots y_{m+1})\cdots p^\alpha(y_{2(m+1)}\vert y_{m+2}\cdots y_{2m+1}),
$$
which, by Condition $(i)$, must be strictly positive. Furthermore, for any two entries $a_{i,j}$ and $a_{k,l}$,
\begin{align} \label{ratio-bounds-2}
\left( \frac{C_L}{C_U} \right)^{\alpha(m+1)}\leq  \frac{a_{i,j}}{a_{k,l}}  \leq \left( \frac{C_U}{C_L} \right)^{\alpha(m+1)}.
\end{align}
Now, for the right eigenvector $v$ of $\widetilde{R}^{(m)}$ corresponding to $\lambda^{(m)}$, we have that,
for any $1 \leq i,j \leq |\mathcal{Y}|^{m+1}$,
$$
\frac{v_i}{v_j}=\frac{\lambda^{(m)} v_i}{\lambda^{(m)} v_j}=\frac{[\widetilde{R}^{(m)} v]_i}{[\widetilde{R}^{(m)} v]_j},
$$
which, together with (\ref{ratio-bounds-2}), implies (\ref{ratio-bounds-1}), as desired.

Now, with (\ref{ratio-bounds-1}) in hand, we infer from (\ref{ineq-1}) that
$$
\lambda^{(m+1)}-\lambda^{(m)}\geq -\lvert\mathcal{Y}\rvert C_F \left( \frac{C_U}{C_L} \right)^{2\alpha} \cdot \left( \rho_F \left(\frac{C_U}{C_L} \right)^{2\alpha} \right)^m.
$$
A parallel argument gives
$$
\lambda^{(m+1)}-\lambda^{(m)} \leq \lvert\mathcal{Y}\rvert C_F \left( \frac{C_U}{C_L} \right)^{2\alpha} \cdot \left( \rho_F \left(\frac{C_U}{C_L} \right)^{2\alpha} \right)^m.
$$
Since $\rho_F < (C_L/C_U)^{2\alpha}$, we obtain
$$
\lvert\lambda^{(m+1)}-\lambda^{(m)}\rvert\leq C_1 \rho^m,
$$
where $C_1=\lvert\mathcal{Y}\rvert C_F (C_U/C_L)^{2\alpha}>0$ and $0<\rho=\rho_F (C_U/C_L)^{2\alpha}<1$. Using the fact that all $\lambda^{(m)}$ are bounded away from $0$ uniformly over $m$ (this follows from Condition $(i)$) and the mean value theorem, we deduce that there exists $C_3 > 0$ such that for all $m$,
$$
\lvert\log\lambda^{(m+1)}-\log\lambda^{(m)}\rvert\leq C_3 \rho^m,
$$
which, by (\ref{Markov-Case}), implies the exponential convergence of $\{H_\alpha(Y^{(m)})\}$, as desired. 
\end{proof}



\begin{rem}
Theorem~\ref{expo-rate} suggests, for $Y$ with small $\rho_F$, a practical method to approximate $H_{\alpha}(Y)$ using $\{H_{\alpha}(Y^{(m)})\}$ (instead of using $\{H_{\alpha, n}(Y)\}$): other than the faster convergence rate, we note that $\{H_{\alpha}(Y^{(m)})\}$ is also easier to compute than $\{H_{\alpha, n}(Y)\}$ due to the fact that $R^{(m)}$ is a sparse matrix, and its largest eigenvalue $\lambda^{(m)}$ can be efficiently computed using the well-known Arnoldi iteration algorithm (see, e.g.,~\cite{tr97}).
\end{rem}


\section{R\' enyi Entropy Rate of General Stationary Ergodic Processes: A Counterexample to Conjecture \ref{conj}}

In this section, we will use the cutting and stacking method~\cite{sh91} to construct a stationary ergodic process such that its \renyi entropy rate of order $\alpha$ does not converge to its Shannon entropy rate as $\alpha$ goes to $1$.


\subsection{The Cutting and Stacking Method} \label{The_Cutting_and_Stacking_Method} \label{SCT_cutting_and_stacking_m}
In this subsection, we give some preliminaries of the cutting and stacking method needed for later sections.
For a more comprehensive exposition of this method, we refer the reader to~\cite{sh96}.

\bigskip

\noindent{\em A. Basic Definitions}
\medskip

Let $\lambda$ be the Lebesgue measure on the real line. A {\em coloum} $\mathcal{C}=\{I_1, I_2, \cdots, I_{h(\mathcal{C})}\}$ is a collection of disjoint subintervals of $[0,1]$ with equal width. We call $I_1$ the {\em base} of $\mathcal{C}$, $I_{h(\mathcal{C})}$ the {\em top} of $\mathcal{C}$ and $h(\mathcal{C})$ the {\em height} of $\mathcal{C}$. Moreover, the {\em width} of $\mathcal{C}$, denoted by $w(\mathcal{C})$, is defined as the width of $I_1$, the {\em support} of $\mathcal{C}$ is defined as $\mbox{supp}(\mathcal{C})\triangleq \cup_{k=1}^h I_i$ and the measure of $\mathcal{C}$, denoted by $\lambda(\mathcal{C})$, is defined as the Lebesgue measure of $\mbox{supp}(\mathcal{C})$.
The columns we consider in this work are often labelled over a finite alphabet. For a 
column $\mathcal{C}=\{I_1, I_2, \cdots, I_{h(\mathcal{C})}\}$, we use $\ell(I_i)$ to denote the label of $I_i$ for any $i$, and use $\ell(\mathcal{C})=\ell(I_1)\ell(I_2)\cdots \ell(I_{h(\mathcal{C})})$ to denote the label of $\mathcal{C}$.

A {\em gadget} $\mathcal{S}=\{\mathcal{C}_1, \mathcal{C}_2\cdots, \mathcal{C}_k\}$ is a collection of columns such that different columns have disjoint supports. Note that the heights of different columns in a gadget are not necessarily the same. If all the columns in a gadget $\mathcal{S}$ have the same height $h$, then the height of the gadget is defined as $h(\mathcal{S})\triangleq h$. The base ({\it resp.} top) of $\mathcal{S}$ is the union of the bases ({\it resp.} top) of all $\mathcal{C}_i$. The width of $\mathcal{S}$ is $w(\mathcal{S})\triangleq \sum_{i=1}^k w(\mathcal{C}_i)$, the support of $\mathcal{S}$ is $\mbox{supp}(\mathcal{S})\triangleq \cup_{i=1}^k \mbox{supp}(\mathcal{C}_i)$ and the measure of $\mathcal{S}$ is $\lambda(\mathcal{S})\triangleq \sum_{i=1}^k \lambda(\mathcal{C}_i)$. The {\em width distribution} of $\mathcal{S}$, denoted by $\boldsymbol{w}(\mathcal{S})$, is a normalized vector
whose $i$-th coordinate is $w(\mathcal{C}_i)/ w (\mathcal{S})$ for any $1\leq i\leq k$, and the {\em measure distribution} of $\mathcal{S}$,  denoted by $\boldsymbol{\lambda}(\mathcal{S})$, is a normalized vector whose $i$-th coordinate is $\lambda(\mathcal{C}_i)/ \lambda(\mathcal{S})$. Finally, a gadget is labelled if all its columns are labelled.

There are two basic operations on columns and gadgets: cutting and stacking. Roughly speaking, a {\em cutting} of a column is an operation that slices the column vertically, resulting in a set of subcolumns; and a {\em stacking} of two columns with equal width is an operation that puts the second column directly onto the first one, which, by definition, results in a single column. A gadget $\mathcal{S}'$ is said to be obtained from another gadget $\mathcal{S}$ via cutting and stacking if each column of $\mathcal{S}'$ is obtained by performing a cutting operation and then a stacking operation on the columns in $\mathcal{S}$. In this paper, we will be mainly concerned with independent cutting and stacking (introduced below), which, as opposed to a general cutting and stacking, is dictated by the width distribution of the gadget.

For a positive integer $m$, a gadget $\mathcal{S}=\{\mathcal{C}_1, \mathcal{C}_2, \cdots, \mathcal{C}_k\}$ can be cut into $m$ {\it copies} $\mathcal{S}_1, \mathcal{S}_2, \cdots, \mathcal{S}_m$ with respect to a distribution $\boldsymbol{\pi}$ via the following two steps:
\begin{enumerate}
\item[1)] For each $j = 1, 2, \cdots, k$, cut the column $\mathcal{C}_j$ into $m$ subcolumns $\mathcal{C}_{j,1}, \mathcal{C}_{j,2}, \cdots, \mathcal{C}_{j,m}$ such that 
$$\left(\frac{\lambda(\mathcal{C}_{j,1})}{\sum_{l} \lambda(\mathcal{C}_{j,l})}, \cdots, \frac{\lambda(\mathcal{C}_{j,m})}{\sum_{l} \lambda(\mathcal{C}_{j,l})}\right)=\boldsymbol{\pi};$$
\item[2)] For each $i = 1, 2, \cdots, m$, let $\mathcal{S}_i \triangleq \{\mathcal{C}_{1,i}, \mathcal{C}_{2,i}, \cdots, \mathcal{C}_{k,i}\}$.
\end{enumerate}
Note from the above definition that each copy of a gadget has the same width distribution as the original gadget (which justifies the use of the word ``copy''). 

We are now ready to introduce the notion of independent cutting and stacking.

\begin{de} \label{inde_cutting_and_stacking} 
Given two disjoint gadgets $\mathcal{S}=\{\mathcal{C}_1, \mathcal{C}_2 \cdots \mathcal{C}_{k_1}\}$ and $\mathcal{S}'=\{\mathcal{C}_1', \mathcal{C}_2' \cdots \mathcal{C}_{k_2}'\}$ with $w(\mathcal{S})=w(\mathcal{S}')$, a new gadget $\mathcal{S}\ast \mathcal{S'}$ is said to be built by applying the {\em independent cutting and stacking} to $\mathcal{S}'$ and $\mathcal{S}$ 
if it is given by the following four steps:
\begin{enumerate}
\item[1)] Cut $\mathcal{S}'$ into $k_1$ copies $\mathcal{S}'_1, \mathcal{S}'_2, \cdots, \mathcal{S}'_{k_1}$ according to $\boldsymbol{w}(\mathcal{S})$, which necessarily implies that $w(\mathcal{S}'_i)=w(\mathcal{C}_i)$ for any $i = 1, 2, \cdots, k_1$. For each $i = 1, 2, \cdots, k_1$, denote 
    $$
    \mathcal{S}'_i\triangleq \{\mathcal{C}_{1,i}', \cdots, \mathcal{C}_{k_2,i}'\};
    $$

\item[2)]  For each $i = 1, 2, \cdots, k_1$, cut $\mathcal{C}_i$ into $k_2$ subcolumns $\mathcal{C}_{i,1}, \cdots, \mathcal{C}_{i,k_2}$ such that for each $j = 1, 2, \cdots, k_2$,
    $$
    w(\mathcal{C}_{i,j}) =w(\mathcal{C}_{j,i}');
    $$

\item[3)] For each $i = 1, 2, \cdots, k_1$ and $j = 1, 2, \cdots, k_2$, put $\mathcal{C}_{j,i}'$ onto $\mathcal{C}_{i,j}$ to form a new column denoted by $\mathcal{C}_{i,j} \ast \mathcal{C}_{j,i}'$;

\item[4)] Finally, let $$\mathcal{S} \ast \mathcal{S'} \triangleq \{ \mathcal{C}_{i,j}\ast \mathcal{C}_{j,i}': 1\leq i\leq k_1, 1\leq j \leq k_2\}.$$
\end{enumerate}
\end{de}

\noindent We refer the reader to Figure~\ref{figure_ind_cutting_and_stacking} for a concrete example on how independent cutting and stacking is done when both $\mathcal{S}$ and $\mathcal{S}'$ have only two columns.

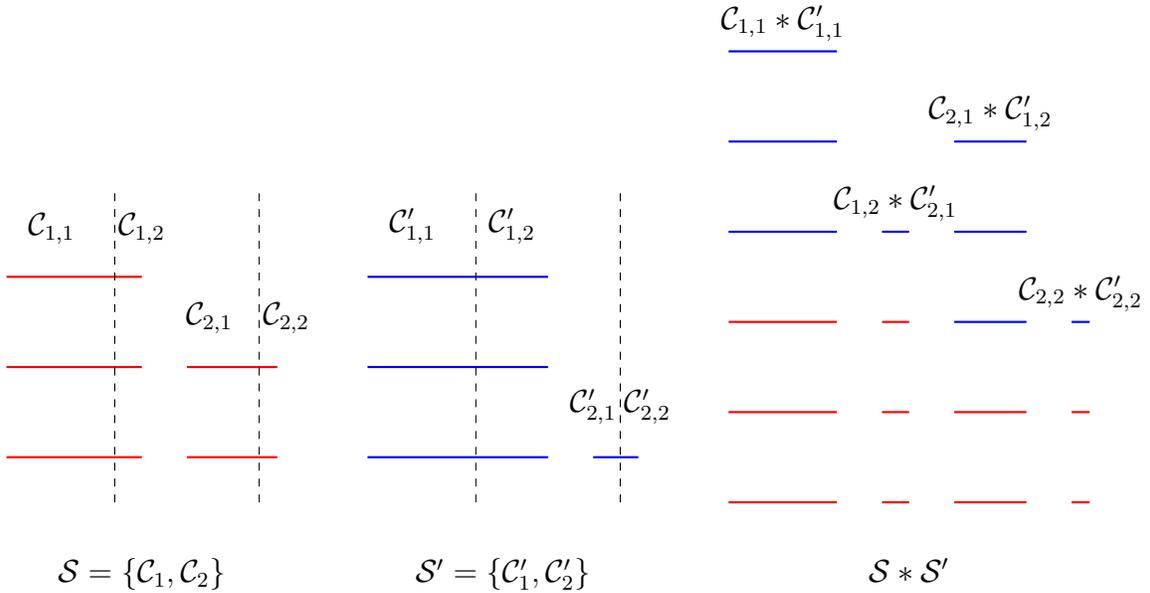
\begin{figure}[htbp!] \label{figure_ind_cutting_and_stacking}
\begin{center}
\begin{tikzpicture}[scale=0.6]

\draw [-, thick, color=red] (0,1) -- (3,1);
\draw [-, thick, color=red] (0,3) -- (3,3);
\draw [-, thick, color=red] (0,5) -- (3,5);
\draw [-, thick, color=red] (4,1) -- (6,1);
\draw [-, thick, color=red] (4,3) -- (6,3);
\path [draw, dashed, color=black] (2.4, 0) -- (2.4, 7);
\path [draw, dashed, color=black] (5.6, 0) -- (5.6, 7);
\node at (3, -1) [coordinate, draw, fill=black, label=below:{$\mathcal{S}=\{\mathcal{C}_1, \mathcal{C}_2\}$}] {};
\node at (1, 5.5) [coordinate, draw, fill=black, label=above:$\mathcal{C}_{1,1}$] {};
\node at (3, 5.5) [coordinate, draw, fill=black, label=above:$\mathcal{C}_{1,2}$] {};
\node at (4.5, 3.5) [coordinate, draw, fill=black, label=above:$\mathcal{C}_{2,1}$] {};
\node at (6.2, 3.5) [coordinate, draw, fill=black, label=above:$\mathcal{C}_{2,2}$] {};

\draw [-, thick, color=blue] (8,1) -- (12,1);
\draw [-, thick, color=blue] (8,3) -- (12,3);
\draw [-, thick, color=blue] (8,5) -- (12,5);
\draw [-, thick, color=blue] (13,1) -- (14,1);
\path [draw, dashed, color=black] (10.4, 0) -- (10.4, 7);
\path [draw, dashed, color=black] (13.6, 0) -- (13.6, 7);
\node at (11, -1) [coordinate, draw, fill=black, label=below:{$\mathcal{S}'=\{\mathcal{C}_1', \mathcal{C}_2'\}$}] {};
\node at (9, 5.5) [coordinate, draw, fill=black, label=above:$\mathcal{C}_{1,1}'$] {};
\node at (11.2, 5.5) [coordinate, draw, fill=black, label=above:$\mathcal{C}_{1,2}'$] {};
\node at (13, 1.5) [coordinate, draw, fill=black, label=above:$\mathcal{C}_{2,1}'$] {};
\node at (14.2, 1.5) [coordinate, draw, fill=black, label=above:$\mathcal{C}_{2,2}'$] {};

\draw [-, thick, color=red] (16,0) -- (18.4,0);
\draw [-, thick, color=red] (16,2) -- (18.4,2);
\draw [-, thick, color=red] (16,4) -- (18.4,4);
\draw [-, thick, color=blue] (16,6) -- (18.4,6);
\draw [-, thick, color=blue] (16,8) -- (18.4,8);
\draw [-, thick, color=blue] (16,10) -- (18.4,10);
\node at (17.2, 10) [coordinate, draw, fill=black, label=above:$\mathcal{C}_{1,1}\ast\mathcal{C}_{1,1}'$] {};

\draw[-, thick, color=red] (19.4,0) -- (20,0);
\draw[-, thick, color=red] (19.4,2) -- (20,2);
\draw[-, thick, color=red] (19.4,4) -- (20,4);
\draw[-, thick, color=blue] (19.4,6) -- (20,6);
\node at (19.7, 6) [coordinate, draw, fill=black, label=above:$\mathcal{C}_{1,2}\ast\mathcal{C}_{2,1}'$] {};

\draw[-, thick, color=red] (21, 0) -- (22.6, 0);
\draw[-, thick, color=red] (21, 2) -- (22.6, 2);
\draw[-, thick, color=blue] (21, 4) -- (22.6, 4);
\draw[-, thick, color=blue] (21, 6) -- (22.6, 6);
\draw[-, thick, color=blue] (21, 8) -- (22.6, 8);
\node at (21.8, 8) [coordinate, draw, fill=black, label=above:$\mathcal{C}_{2,1}\ast\mathcal{C}_{1,2}'$] {};

\draw[-, thick, color=red] (23.6, 0) -- (24, 0);
\draw[-, thick, color=red] (23.6, 2) -- (24, 2);
\draw[-, thick, color=blue] (23.6, 4) -- (24, 4);
\node at (20, -1) [coordinate, draw, fill=black, label=below:$\mathcal{S}\ast \mathcal{S}'$] {};
\node at (23.8, 4) [coordinate, draw, fill=black, label=above:$\mathcal{C}_{2,2}\ast\mathcal{C}_{2,2}'$] {};

\end{tikzpicture}
\end{center}
\caption{This figure illustrates how to apply independent cutting and stacking to two gadgets $\mathcal{S}=\{\mathcal{C}_1, \mathcal{C}_2\}$ and $\mathcal{S}'=\{\mathcal{C}_1', \mathcal{C}_2'\}$, where cutting a column is represented by a dashed line.}
\label{figure_ind_cutting_and_stacking}
\end{figure}

One of the important properties of the independent cutting and stacking of two gadgets is that the width distribution of the resulting gadget is the Kronecker product of the width distributions of the original two gadgets, detailed below.

\begin{pr}  \label{width_dis} \textnormal{\cite{sh96}}
Let $\mathcal{S}$ and $\mathcal{S}'$ be two gadgets with the same width. Then we have
$$
\boldsymbol{w}(\mathcal{S}\ast \mathcal{S}') = \boldsymbol{w}(\mathcal{S}) \otimes \boldsymbol{w}(\mathcal{S}'),
$$
where for any two vectors $u$ and $v$, $u \otimes v$ denotes the Kronecker product of $u$ and $v$. Moreover, if the heights of $\mathcal{S}$ and $\mathcal{S}'$ are both well-defined, then
$$
\boldsymbol{\lambda}(\mathcal{S}\ast \mathcal{S}') = \boldsymbol{\lambda}(\mathcal{S}) \otimes \boldsymbol{\lambda}(\mathcal{S}').
$$
\end{pr}
\medskip

As detailed in the following definition, independent cutting and stacking of two gadgets as in Definition~\ref{inde_cutting_and_stacking} can be iteratively applied and composed to give rise to a multi-fold version for a single gadget.
\begin{de} \label{M_fold_icas}
Let $M$ be a positive integer and $\mathcal{S}$ be a gadget that can be cut into $M$ identical copies $\{\mathcal{S}_1, \mathcal{S}_2, \cdots, \mathcal{S}_M\}$. Then the new gadget $\mathcal{S}^{\langle M \rangle}\triangleq \mathcal{S}_1\ast \mathcal{S}_2  \ast \cdots \ast \mathcal{S}_{M}$ is said to be the gadget obtained by applying the {\em $M$-fold independent cutting and stacking} to $\mathcal{S}$, where for any $3\leq i \leq M$, $\mathcal{S}_1\ast \mathcal{S}_2  \ast \cdots \ast \mathcal{S}_{i-1}\triangleq (\mathcal{S}_1\ast \mathcal{S}_2  \ast \cdots \mathcal{S}_{i-2}) \ast \mathcal{S}_{i-1}$. {In the sequel, we sometimes call $\mathcal{S}^{\langle M \rangle}$
the $M$-fold independent cutting and stacking of $\mathcal{S}$ for simplicity.}
\end{de}

\begin{rem}
A direct application of Proposition \ref{width_dis} indicates that the width distribution ({\em resp.} measure distribution) of 
$\mathcal{S}^{\langle M \rangle}$ is the $M$-fold Kronecker product
of the width distribution ({\em resp.} measure distribution) of $\mathcal{S}$.
\end{rem}
\bigskip

\noindent{\em B. Processes and Gadgets}
\medskip

Given a probability space $(\Omega, \mathcal{F}, P)$, a partition $\mathcal{P} \triangleq\{\Omega_1, \cdots, \Omega_A\}$ of the sample space $\Omega$ and a transformation $T: \Omega \rightarrow \Omega$ that is well defined almost everywhere in $\Omega$, we can define a random process $\{X_n\}_{n=1}^\infty$ over the alphabet $\{1,2,\cdots ,A\}$ via the following two steps:

\begin{enumerate}
\item[(1)] for each $\omega\in \Omega$ such that $T\omega$ is not well-defined, define $T \omega$ to be an arbitrary point inside $\Omega$;

\item[(2)] for any $\omega\in \Omega$ and any positive integer $n$, let
\begin{align} \label{T,P}
X_n(\omega)\triangleq a \qquad \mbox{if }  T^{n-1} \omega \in \Omega_a , a\in \{1,2, \cdots A\}.
\end{align}
\end{enumerate}
Evidently, the process $\{X_n\}_{n=1}^\infty$ is determined by the transformation $T$ and the partition $\mathcal{P}$ and therefore will be referred to as a $(T,\mathcal{P})$-process in the sequel. Here we remark that since the set of points on which $T$ is not well-defined has Lebesgue measure $0$, the choice of $T\omega$ as in Step ($1$) has no influence on the distribution of the process $\{X_n\}_{n=1}^\infty$.

In order to obtain a process from a gadget in a similar way as above, we first need the following definition.
\begin{de} \label{par_map_def}
Let $\mathcal{S}$ be a gadget labelled over a finite alphabet $\mathcal{A}$. The partition 
$$\mathcal{P}_\mathcal{S} \triangleq \{\mathcal{P}_a: a\in \mathcal{A}\}$$
is called the {\em partition induced by $\mathcal{S}$}, where $\mathcal{P}_a$ is the union of all the levels in $\mathcal{S}$ labelled by $a$. We further use $T_\mathcal{S}$ to denote the induced map of $\mathcal{S}$, which maps any point that is not in the top of $\mathcal{S}$ directly upwards (see Figure 2). 
\end{de}
\begin{figure}[htbp!] \label{partition_and_map}
\begin{center}
\begin{tikzpicture}[scale=0.6]

\draw [-, thick, color=red] (8,1) -- (11,1); 
\draw [-, thick, color=red] (8,3) -- (11,3); 
\draw [-, thick, color=red] (8,5) -- (11,5); 
\draw [-, thick, color=red] (12,1) -- (14,1); 
\draw [-, thick, color=red] (12,3) -- (14,3); 
\path [draw, ->, thick, color=black] (9, 3) -- (9, 5); 
\path [draw, ->, thick, color=black] (10, 1) -- (10, 3); 
\path [draw, ->, thick, color=black] (13, 1) -- (13, 3); 
\node at (11, 0) [coordinate, draw, fill=black, label=below:{$\mathcal{S}$}] {}; 
\node at (10.2, 1) [coordinate, draw, fill=black, label=below:{$\omega_1$}] {}; 
\node at (10.8, 3) [coordinate, draw, fill=black, label=below:{$T_{\mathcal{S}}\omega_1$}] {}; 
\node at (8.8, 3) [coordinate, draw, fill=black, label=below:{$\omega_2$}] {}; 
\node at (8.2, 5) [coordinate, draw, fill=black, label=below:{$T_{\mathcal{S}}\omega_2$}] {}; 
\node at (13.2, 1) [coordinate, draw, fill=black, label=below:{$\omega_3$}] {}; 
\node at (13.8, 3) [coordinate, draw, fill=black, label=below:{$T_{\mathcal{S}}\omega_3$}] {}; 
\end{tikzpicture}
\end{center}
\caption{$T_\mathcal{S}$ maps any point in $\mathcal{S}$ that is not in the top directly upwards.}
\end{figure}
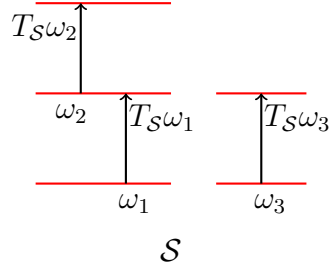

\begin{rem}
$T_\mathcal{S}$ is a Lebesgue measure-preserving map and it is not well-defined on the top of $\mathcal{S}$. 
\end{rem}
\begin{rem} \label{cutting_partition_map}
Let $\mathcal{S}$ be a labelled gadget, $\mathcal{S}'$ be a cutting and stacking of $\mathcal{S}$, and $\mathcal{P}_{\mathcal{S}}$, $\mathcal{P}_{\mathcal{S}'}$ be the partition induced by $\mathcal{S}$ and $\mathcal{S}'$, respectively. Then $\mathcal{P}_{\mathcal{S}}=\mathcal{P}_{\mathcal{S}'}$ 
and $T_{\mathcal{S}'}$ extends $T_\mathcal{S}$, since the top of $\mathcal{S}'$ is a subset of the top of $\mathcal{S}$.  
\end{rem}

Now, consider the probability space $([0, 1], \mathcal{B}, \lambda)$, where $[0,1]$ is the unit interval, $\mathcal{B}$ is the Borel $\sigma$-algebra on $[0,1]$ and $\lambda$ is the Lebesgue measure. To obtain a $(T, \mathcal{P})$-process from a gadget, we need to find a measure-preserving map on $[0,1]$ such that it is well-defined almost everywhere (note that any map induced by a single gadget is necessarily not well-defined on a set of positive measure). To this end, we start with a  gadget $\mathcal{S}(1)$ with support $[0,1]$ and induced map $T_{\mathcal{S}(1)}$. Applying cutting and stacking to $\mathcal{S}(1)$ gives a new gadget $\mathcal{S}(2)$ such that its induced map $T_{\mathcal{S}(2)}$ is an extension of $T_{\mathcal{S}(1)}$. Continuing in this way we obtain a sequence of gadgets $\{\mathcal{S}(m)\}_{m=1}^\infty$ such that for each $m\geq 1$, $\mathcal{S}(m+1)$ is a cutting and stacking of $\mathcal{S}(m)$ and $T_{\mathcal{S}(m+1)}$ is an extension of $T_{\mathcal{S}(m)}$. If the measure of the top of $\mathcal{S}(m)$ goes to $0$, then Remark \ref{cutting_partition_map} implies that $\{T_{\mathcal{S}(m)}\}$ has a common extension that is well-defined almost everywhere on $[0,1]$, and therefore the corresponding $(T,\mathcal{P})$-process is well-defined. These ideas are summarized in the following theorem.
\begin{thm} \label{pro_of_gad}\textnormal{\cite{sh96}}
Let $\{\mathcal{S}(m)\}_{m=1}^\infty$ be a sequence of labeled gadgets with the following properties:
\begin{enumerate}
\item[(1)] $\lambda(\mathcal{S}(1))=1$;
\item[(2)] For any $m\geq 1$, $\mathcal{S}(m+1)$ is a
cutting and stacking (not necessarily an independent cutting and stacking) of $\mathcal{S}(m)$; 
\item[(3)] $w(\mathcal{S}(m))$ goes to $0$ as $m$ goes to infinity.
\end{enumerate}
Then $\{T_{\mathcal{S}(m)}\}_{m=1}^\infty$ has a common extension $T$ which is well-defined on $[0,1]$ almost everywhere. Consequently, $\{\mathcal{S}(m)\}_{m=1}^\infty$ defines a $(T, \mathcal{P}_{\mathcal{S}(1)})$-process, where $\mathcal{P}_{\mathcal{S}(1)}$ is the partition induced by $\mathcal{S}(1)$.
\end{thm}

The $(T, \mathcal{P}_{\mathcal{S}(1)})$-process given by $\{\mathcal{S}(m)\}_{m=1}^\infty$ as in Theoerm \ref{pro_of_gad} is called the {\em final process} of $\{\mathcal{S}(m)\}_{m=1}^\infty$. Note that $T$ depends on the sequence $\{\mathcal{S}(m)\}_{m=1}^\infty$ rather than any single element thereof.

The following theorem characterizes the finite-dimensional distribution of the final process.

\begin{thm} \label{joint_dis} \textnormal{\cite{sh96}}
Let $\{\mathcal{S}(m)\}_{m=1}^\infty$ be a sequence of gadgets labelled over the alphabet $\mathcal{A}$ satisfying the conditions in Theorem \ref{pro_of_gad}. Let $\mu$ denote the distribution of the final process given by $\{\mathcal{S}(m)\}_{m=1}^\infty$. Then
$$
\mu(a_1^k)=\lim_{m\rightarrow \infty} \sum_{\mathcal{C}\in \mathcal{S}(m)} p_k(a_1^k\vert\mathcal{C}) \lambda(\mathcal{C}) \qquad \mbox{for any $a_1^k\in \mathcal{A}^k$,}
$$
where $p_k(a_1^k|\mathcal{C})$ is defined as
$$
p_k(a_1^k\vert\mathcal{C})\triangleq \frac{\lvert\{1\leq i\leq h(\mathcal{C})-k+1: \ell(\mathcal{C})_i^{i+k-1}=a_1^k\}\rvert}{h(\mathcal{C})-k+1}
$$
and $\ell(\mathcal{C})_i^{i+k-1}$ is the subsequence consists of symbols from the $i$-th position to the $(i+k-1)$-th position of $\ell(\mathcal{C})$.
\end{thm}

We now pay our attention to the ergodicity of the final process. 

\begin{de}
Two gadgets $\mathcal{S}$ and $\mathcal{S}'$ are said to be {\em $\epl$-independent} if
$$\sum_{\mathcal{C}\in \mathcal{S}} \sum_{\mathcal{D}\in \mathcal{S}'} \lvert\lambda(\mathcal{C}\cap \mathcal{D}) -\lambda(\mathcal{C}) \lambda(\mathcal{D})\rvert\leq \epl,$$
where $\mathcal{C} \cap \mathcal{D}\triangleq \mbox{supp}(\mathcal{C})\cap \mbox{supp}(\mathcal{D})$.
\end{de}

The following theorem from \cite{sh96} gives a sufficient condition for the final process to be ergodic.

\begin{thm} \label{cut_and_erg} \textnormal{\cite{sh96}}
Let $\{\mathcal{S}(m)\}_{m=1}^\infty$ be a sequence of gadgets satisfying the following conditions:
\begin{enumerate}
\item[(1)] For any $m\geq 1$, $\mathcal{S}(m+1)$ is obtained by performing cutting and stacking on $\mathcal{S}(m)$;

\item[(2)] $\lambda(\mathcal{S}(m))\rightarrow 1$ and $w(\mathcal{S}(m))\rightarrow 0$ as $m\rightarrow \infty$;

\item [(3)] There is a sequence $\{\epl_m\}$ with $\lim_{m\rightarrow \infty} \epl_m=0$ such that for any $m\geq 1$, $\mathcal{S}(m)$ and $\mathcal{S}(m+1)$ are $\epl_m$-independent.
\end{enumerate}
Then the final process given by $\{\mathcal{S}(m)\}_{m=1}^\infty$ is ergodic.
\end{thm}


\begin{rem} \label{ind_erg}
According to Theorem I.10.11 of \cite{sh96}, 
for any $\epl_m$, there is an $M_m$ such that 
 $\mathcal{S}(m)$ and $\mathcal{S}(m+1)$ are $\epl_m$-independent, where $\mathcal{S}(m+1)\triangleq\mathcal{S}(m)^{\langle M_m \rangle}$.
This result plays an important role in the remainder of this paper.
\end{rem}

\subsection{Construction of the Counterexample} \label{constru_of_counter}

In this section, we construct a stationary and ergodic process $\mu$~\footnote{Throughout this section, for for a random process $Z$ with distribution $\mu$, we may use $Z$ and $\mu$ interchangeably to denote the process.} for which we will prove in subsequent sections that its R\' enyi entropy rate $H_{\alpha}(\mu)$ exists for all $\alpha\in [1,\infty)$ yet fails to converge to $H(\mu)$ as $\alpha$ monotonically decreasing to $1$. As mentioned before, such a counterexample is easy to construct without the ergodicity assumption (see Section IV of \cite{ra2001}); on the other hand, although it has been shown in \cite{NH2008} that $H_{\alpha}(\mu)$ always converges to $H(\mu)$ when $\mu$ is weakly $\psi$-mixing (hence ergodic), it remains unknown that whether the same result is true if $\mu$ is assumed to be as general as ergodic. In this and the following subsections, we give a negative answer to this question by constructing a counterexample. The idea of the construction is to use the cutting and stacking method to control finite-dimensional probabilities of the process.

We need the following definition before constructing the counterexample. Roughly speaking, it defines a ``fractional" version of the $M$-fold cutting and stacking.
\begin{de} \label{fractional_ind}
Let $\mathcal{S}=\{\mathcal{S}_L , \mathcal{S}_R\}$ be a gadget with measure $1$ where $\mathcal{S}_L$ is a single column with $\lambda(\mathcal{S}_L)=\alpha$.
We use $\{ \langle \mathcal{S}_L\rangle_M, \mathcal{S}_R^{\langle M \rangle} \}$ to denote a gadget obtained from $\mathcal{S}$ by applying the 
following steps:
\begin{enumerate}
\item[1)] $\langle \mathcal{S}_L\rangle_{M}$ is obtained from $\mathcal{S}_L$ by cutting $\mathcal{S}_L$ evenly into $M$ subcolumns and then stack them into a single column; 
\item[2)] $\mathcal{S}_R^{\langle M \rangle}$ is obtained by applying the $M$-fold cutting and stacking to $\mathcal{S}_R$.
\end{enumerate}
\end{de}



Making use of the above definition, we now elaborate the construction of our counterexample.
We first construct a sequence of gadgets $\{\mathcal{G}(m)\}_{m=1}^\infty$, each of which has measure $1$ and labelled over the alphabet $\{0,1\}$ through the following steps:
\medskip

\textbf{Step 1:} Choose two sequences 
of constants $\{\alpha_m\}_{m=1}^\infty$ and $\{\beta_m\}_{m=1}^\infty$ such that $0<\alpha_m< \beta_m<1$ for any feasible $m$, $\lim_{m\rightarrow \infty}\alpha_m=0$ and $\lim_{m\rightarrow \infty}\beta_m=0$. Also choose a strictly positive integer $l_1$. (The choices of $\{\alpha_m\}_{m=1}^\infty, \{\beta_m\}_{m=1}^\infty$ and $l_1$ will be specified later in the proof of Proposition \ref{property_of_mu})
\medskip

\textbf{Step 2:} Let $m=1$ and define $\mathcal{G}(1)\triangleq \{\mathcal{L}(1), \mathcal{R}(1)\}$, where $\mathcal{L}(1)$ is a single column of height $l_1$, width $\frac{1}{l_1 2^{2 l_1 /3}}$ and label $\underbrace{1 1 \cdots 1}_{l_1}$, and $\mathcal{R}(1)$ consists of $2^{2l_1/3} -1$ columns, each of which has height $l_1$ and width $\frac{1}{l_1 2^{2 l_1 /3 }}$. We further assign distinct labels to columns in $\mathcal{R}(1)$ such that 
the last column is labeled $\underbrace{1 1 \cdots 1}_{l_1}$. 
Note that 
$\lambda(\mathcal{G}(1))=1$ and 
$\lambda(\mathcal{L}(1))=\beta_1$.
\medskip

\textbf{Step 3:} Suppose $l_m$ has been chosen and $\mathcal{G}(m)\triangleq \{\mathcal{L}(m), \mathcal{R}(m)\}$ has already been constructed where $\lambda(\mathcal{L}(m))=\beta_m$. Cut $\mathcal{L}(m)$ into two copies $\mathcal{L}(m,1)$ and $\mathcal{L}(m,2)$ such that 
$\lambda(\mathcal{L}(m,1))=\beta_{m+1}$ and 
$\lambda(\mathcal{L}(m,2))=\beta_m-\beta_{m+1}.$
\medskip

\textbf{Step 4:} Choose a positive integer $l_{m+1}$ (the existence of $l_{m+1}$ follows from Remark \ref{ind_erg} and the fact that $\beta_m>\alpha_m$) large enough such that the following three conditions hold:
\begin{enumerate}
\item [$(a)$] $\displaystyle\frac{l_{m+1}}{l_m}$ is a positive integer;
\item [$(b)$] $\displaystyle\left( 1-\frac{m}{l_m}\right) \beta_m\geq \alpha_m$;
\item [$(c)$] Let $\mathcal{R}(m+1)\triangleq \{\mathcal{L}(m,2), \mathcal{R}(m)\}^{\langle l_{m+1}/l_m\rangle}$.
Then $\mathcal{R}(m+1)$ and $\{\mathcal{L}(m,2) , \mathcal{R}(m)\}$ are $\epl_m$-independent.
\end{enumerate}
\medskip

\textbf{Step 5:} Define $\mathcal{L}(m+1)\triangleq\langle \mathcal{L}(m,1) \rangle_{l_{m+1}/l_m}$ and $\mathcal{G}(m+1)\triangleq \{\mathcal{L}(m+1), \mathcal{R}(m+1)\}$. 
\medskip

\textbf{Step 6:} Increase the value of $m$ by $1$ and go to Step 3.
\medskip

It is obvious from the above construction that for any $m\geq 1$, $\lambda(\mathcal{G}(m))=1$, and $w(\mathcal{G}(m))$ converges to $0$ as $m$ goes to infinity. Hence, letting $T$ denote the common extension of the maps  $\{T_{\mathcal{G}(m)}\}_{m=1}^\infty$ and $\mathcal{P}_{\mathcal{G}}$ be the partition induced by $\mathcal{G}(1)$, we infer from Theorem \ref{pro_of_gad} that the $(T, \mathcal{P}_{\mathcal{G}})$-process of $\{\mathcal{G}(m)\}_{m=1}^\infty$ is well defined. By properly choosing $l_1$ and $\beta_m$ (the choices will be specified in the next section), this binary process will be the counterexample we construct.

For notational simplicity, in the remainder of this paper, we use $\mu_\infty^\mathcal{G}$ to denote the final process constructed by the above steps. 


\subsection{Properties of $\mu_\infty^\mathcal{G}$} \label{Properties_of_mu}

In this section, we will show that with proper choices of constants $l_1$ and $\{\beta_m\}_{m=1}^\infty$, the final process $\mu_\infty^{\mathcal{G}}$ has some desirable properties, which, as will be shown later, are essential for establishing that $H_{\alpha}(\mu_\infty^{\mathcal{G}})$ does not converge to $H(\mu_\infty^{\mathcal{G}})$ as $\alpha$ monotonically decreasing to $1$. More specifically, we will establish the following proposition.
\begin{pr} \label{property_of_mu}
The constant $l_1$ and sequences $\{\alpha_m\}_{m=1}^\infty, \{\beta_m\}_{m=1}^\infty$ in Section~\ref{constru_of_counter} can be chosen such that the followings hold:
\begin{enumerate}
\item [(A)] $\lim\limits_{m\rightarrow \infty} \displaystyle\frac{1}{m} \log \alpha_m =0$;
\item [(B)] $\mu_\infty^\mathcal{G}(\{x_1^\infty: x_1^m=\underbrace{1\cdots 1}_{m}\}) \geq \alpha_m;$
\item [(C)]$\mu_\infty^\mathcal{G}$ is ergodic;
\item [(D)]$H(\mu_\infty^{\mathcal{G}}) > 1/2$.
\end{enumerate}
\end{pr}

Some discussions are needed before proving this proposition. The proofs of Properties $(A)$, $(B)$ and $(C)$ in Proposition~\ref{property_of_mu} follow similar arguments as in Section III.1.c of \cite{sh96} and are relatively easy. The proof of Property $(D)$ is however somewhat subtle, as it requires not only a specific description on how a ``fractional" independent cutting and stacking affects the width distribution, but also an explicit relationship between the Shannon entropy rate of the final process and the width distributions of $\{\mathcal{G}(m)\}_{m=1}^\infty$. To this end, we begin with some definitions.

\begin{de} \label{label_function}
Let $\mathcal{S}=\{\mathcal{C}_1,\mathcal{C}_2, \cdots, \mathcal{C}_k\}$ be a gadget labelled over the alphabet $\mathcal{A}=\{1,2,\cdots A\}$ such that $\lambda(\mathcal{S})=1$ and for any $1\leq i\leq k$, $h(\mathcal{C}_i)=h$ for some constant $h$. 
Then the {\em normalized Shannon entropy of $\mathcal{S}$} is defined as
$$H({\mathcal{S}}) \triangleq -\frac{1}{h} \sum_{a_1^h \in \mathcal{A}_1^h} \lambda_\mathcal{S}(a_1^h) \log \lambda_{\mathcal{S}}(a_1^h),$$
where $\lambda_{\mathcal{S}} (a_1^h) \triangleq \sum_{i=1}^k \lambda(\mathcal{C}_i) 1_{\{\ell(\mathcal{C}_i)=a_1^h\}}
$ and for any $1\leq i\leq k$, $\ell(\mathcal{C}_i)$ denotes the label of the $\mathcal{C}_i$. 
In particular, if $\mathcal{C}_i$'s have distinct labels, then
$$H(\mathcal{S})=-\frac{1}{h} \sum_{i=1}^k \lambda(\mathcal{C}_i) \log \lambda(\mathcal{C}_i).$$
\end{de}
It is clear from this definition that columns with the same label have to be ``merged" when computing the normalized Shannon entropy of a gadget. Hence, we introduce the following definition.
\begin{de}
For any $k\geq 1$, let $\mathcal{C}_1=\{I_1^{(1)},\cdots, I_h^{(1)}\}, \cdots, \mathcal{C}_k=\{I_1^{(k)},\cdots, I_h^{(k)}\}$ be $k$ columns with the same height $h$ and the same label. The column $\mathcal{D}\triangleq\{\cup_{i=1}^k I_1^{(i)}, \cdots, \cup_{i=1}^k I_h^{(i)} \}$ is called the {\em merging} of $\mathcal{C}_1, \mathcal{C}_2, \cdots, \mathcal{C}_k$. 
One can easily verify
$$
h(\mathcal{D})=h,\quad w(\mathcal{D})= \sum_{i=1}^k w(\mathcal{C})
$$
and for any $1\leq i\leq k$, 
$$
\mbox{supp}(\mathcal{D})=\bigcup_{i=1}^k \mbox{supp}(\mathcal{C}_i), \quad \ell(\mathcal{D})=\ell(\mathcal{C}_i) .
$$
Moreover,
 a gadget $\mathcal{S}'$ is called {\em a merging of another gadget} $\mathcal{S}$ if columns in $\mathcal{S}'$ have distinct labels and each column in $\mathcal{S}'$ is a merging of some columns in $\mathcal{S}$.
\end{de}

\begin{rem} \label{re_2.9}
It is immediate from this definition that the normalized Shannon entropy of a gadget will not change if we cut or merge some columns of this gadget. 
\end{rem}

In order to prove Proposition \ref{property_of_mu}, we will tackle the following two questions regarding the gadgets $\{\mathcal{G}(m)\}_{m=1}^\infty$ and the final process $\mu_\infty^{\mathcal{G}}$ constructed in Section \ref{constru_of_counter}: 1) what is the relationship between the normalized Shannon entropy of $\mathcal{G}(m)$ and $\mathcal{G}(m+1)$; 2) what is the relationship between the Shannon entropy of the sequence $\{\mathcal{G}(m)\}_{m=1}^\infty$ and the Shannon entropy rate of $\mu_\infty^{\mathcal{G}}$? 
We begin with the following lemma.


\begin{lem} \label{entropy_rate_of_gadget_2}
Let $\mathcal{S}=\{\mathcal{C}_0, \mathcal{C}_1, \cdots, \mathcal{C}_k\}$ be a gadget labelled over the alphabet $\mathcal{A}$ such that 
\begin{enumerate}
\item[(1)] $\lambda(\mathcal{S})=1$; 
\item[(2)] for any $0\leq i \leq k$,
$h(\mathcal{C}_i)=h$ for some $h\in \mathbb{N}^+$; 
\item[(3)] there exists $1\leq i'\leq k$ such that $\ell(\mathcal{C}_{i'})=\ell(\mathcal{C}_0)$. 
\end{enumerate}
Also let $\widetilde{\mathcal{S}}\triangleq \{\mathcal{C}_0, \widetilde{\mathcal{C}}_1, \cdots, \widetilde{\mathcal{C}}_{k'}\}$, where $\{\widetilde{\mathcal{C}}_1, \cdots, \widetilde{\mathcal{C}}_{k'}\}$ is the merging of $\{\mathcal{C}_1, \cdots, \mathcal{C}_k\}$. Define $\mathcal{S}'\triangleq \{\langle\mathcal{C}_0 \rangle_M, \{\mathcal{C}_1, \cdots, \mathcal{C}_k\}^{\langle M \rangle}\}$ and $\widetilde{\mathcal{S}}' \triangleq \{\langle\mathcal{C}_0 \rangle_M, \{\widetilde{\mathcal{C}}_1, \cdots, \widetilde{\mathcal{C}}_{k'}\}^{\langle M \rangle}\}$.
Then,
$$H({\mathcal{S}'})=H({\tilde{\mathcal{S}}'}).$$
\end{lem}


\begin{proof}
Applying merging if necessary, we can assume without loss of generality that among $\mathcal{C}_1, \cdots, \mathcal{C}_k$, only $\mathcal{C}_{k-1}$ and $\mathcal{C}_k$ have the same label as $\mathcal{C}_0$ and furthermore all other columns have distinct labels. 

Under the above assumption, we have 
\begin{align*}
\widetilde{\mathcal{S}}=\{\mathcal{C}_0, \mathcal{C}_1, \cdots,\mathcal{C}_{k-2}, \widetilde{\mathcal{C}}_{k-1}\} \quad \mbox{and} \quad \widetilde{\mathcal{S}}' \triangleq \{\langle\mathcal{C}_0 \rangle_M, \{\mathcal{C}_1, \cdots, \mathcal{C}_{k-2},\widetilde{\mathcal{C}}_{k-1}\}^{\langle M \rangle}\}
\end{align*} 
where $\widetilde{\mathcal{C}}_{k-1}$ is the merging of $\mathcal{C}_{k-1}$ and $\mathcal{C}_k$ (this implies $\lambda(\widetilde{\mathcal{C}}_{k-1})=\lambda(\mathcal{C}_{k-1})+\lambda(\mathcal{C}_{k})$).
We first note from the definition of cutting and stacking that the label of any column $\mathcal{C}\in \mathcal{S}' \cup \widetilde{\mathcal{S}}'$ is given by
$$\ell(\mathcal{C})=\boldsymbol{b}_{1} \boldsymbol{b}_{2} \cdots \boldsymbol{b}_{M},$$ 
where for any $1\leq j\leq M$, $\boldsymbol{b}_{j}\in\{\ell(\mathcal{C}_0),\ell(\mathcal{C}_1) \cdots, \ell(\mathcal{C}_{k-2})\}$ (note that $\boldsymbol{b}_j$ is a block of length $h(\mathcal{S})$.
Then, to prove $H({\mathcal{S}'}) = H({\widetilde{\mathcal{S}}'})$, it suffices to show that
\begin{equation} \label{same_measure}
\lambda_{\mathcal{S}'}(\boldsymbol{b}_{1} \boldsymbol{b}_{2} \cdots \boldsymbol{b}_{M})=\lambda_{\widetilde{\mathcal{S}}'}(\boldsymbol{b}_{1} \boldsymbol{b}_{2} \cdots \boldsymbol{b}_{M})
\end{equation}
for any sequence $\boldsymbol{b}_{1} \boldsymbol{b}_{2} \cdots \boldsymbol{b}_{M}\in \{\ell(\mathcal{C}_0),\ell(\mathcal{C}_1), \cdots, \ell(\mathcal{C}_{k-2})\}^{M}$, where $\lambda_\mathcal{S}$ and $\lambda_{\mathcal{S}'}$ are defined in Definition \ref{label_function}.

To this end, we consider two cases of $\boldsymbol{b}_{1} \boldsymbol{b}_{2} \cdots \boldsymbol{b}_{M}$:

\noindent{\em Case 1:} $\boldsymbol{b}_{j}= \ell(\mathcal{C}_0)$ for any $1\leq j \leq M$.

First note that for a column in $\mathcal{S}'$ ({\em resp.} $\widetilde{\mathcal{S}}'$) with the label $\underbrace{\ell(\mathcal{C}_0) \cdots \ell(\mathcal{C}_0)}_{M}$, it is either the column $\langle \mathcal{C}_0 \rangle_M$ or a column in $\{\mathcal{C}_1, \cdots, \mathcal{C}_k\}^{\langle M \rangle}$({\em resp.} $\{\mathcal{C}_1, \cdots, \mathcal{C}_{k-2},\widetilde{\mathcal{C}}_{k-1}\}^{\langle M \rangle}$). Hence, making use of Proposition \ref{width_dis}, we have 
\begin{align} \label{Case1Lc0S}
&\lambda_{\mathcal{S}'}(\underbrace{\ell(\mathcal{C}_0) \cdots \ell(\mathcal{C}_0)}_{M}) \notag \\
&= \lambda(\langle\mathcal{C}_0\rangle_M)+(1-\lambda(\mathcal{C}_0))\left( \sum_{\substack{i_1, \cdots, i_M: \\\ell(\mathcal{C}_{i_1})\cdots \ell(\mathcal{C}_{i_M})=w_1 \cdots w_M }} \frac{\lambda(\mathcal{C}_{i_1})}{1-\lambda(\mathcal{C}_0)}\cdots \frac{\lambda(\mathcal{C}_{i_M})}{1-\lambda(\mathcal{C}_0)}\right) \notag \\
&= \lambda(\langle\mathcal{C}_0\rangle_M)+(1-\lambda(\mathcal{C}_0))\left( \sum_{i_1 \cdots i_M\in \{k-1, k\}^M  } \frac{\lambda(\mathcal{C}_{i_1})}{1-\lambda(\mathcal{C}_0)}\cdots \frac{\lambda(\mathcal{C}_{i_M})}{1-\lambda(\mathcal{C}_0)}\right) \notag \\
&=\lambda(\langle\mathcal{C}_0\rangle_M)+ (1-\lambda(\mathcal{C}_0)) \left(\frac{\lambda(\mathcal{C}_{k-1})+\lambda(\mathcal{C}_k)}{1-\lambda(\mathcal{C}_0)}\right)^M
\end{align}
and
\begin{align} \label{Case1Lc0tildeS}
\lambda_{\widetilde{\mathcal{S}}'}(\underbrace{\ell(\mathcal{C}_0) \cdots \ell(\mathcal{C}_0)}_{M}) = \lambda(\langle \mathcal{C}_0\rangle_M)+ (1-\lambda(\mathcal{C}_0)) \left(\frac{\lambda(\widetilde{\mathcal{C}}_{k-1})}{1-\lambda(\mathcal{C}_0)}\right)^M.
\end{align}
Since 
$\lambda(\mathcal{C}_{k-1})+\lambda(\mathcal{C}_k)=\lambda(\widetilde{\mathcal{C}}_{k-1})$, we infer from (\ref{Case1Lc0S}) and (\ref{Case1Lc0tildeS}) that 
$$
\lambda_{\mathcal{S}'}(\underbrace{\ell(\mathcal{C}_0) \cdots \ell(\mathcal{C}_0)}_{M})=\lambda_{\widetilde{\mathcal{S}}'}(\underbrace{\ell(\mathcal{C}_0) \cdots \ell(\mathcal{C}_0)}_{M}).
$$

\noindent{\em Case 2:} $\boldsymbol{b}_{j}\neq \ell(\mathcal{C}_k)$ for some $1\leq j\leq M$.

Note that in this case, any column in $\mathcal{S}'$ ({\em resp.} $\widetilde{\mathcal{S}}'$) with the label $\boldsymbol{b}_{1}\boldsymbol{b}_{2}\cdots \boldsymbol{b}_{M}$ must be in  $\{\mathcal{C}_1, \cdots, \mathcal{C}_k\}^{\langle M \rangle}$({\em resp.} $\{\mathcal{C}_1, \cdots, \mathcal{C}_{k-2},\widetilde{\mathcal{C}}_{k-1}\}^{\langle M \rangle}$). Define the set 
$E\triangleq \{1\leq j\leq k:  \boldsymbol{b}_j=\ell(\mathcal{C}_0)\}$ and ${E}^c\triangleq \{1,2,\cdots M\} \setminus E$. Since $\mathcal{C}_1, \cdots, \mathcal{C}_{k-2}$ have distinct labels, we see that for any $j\in E^c$, there is only one column labelled $\boldsymbol{b}_{j}$. Let $f: E^c \rightarrow \{1, 2, \cdots k-2\}$ be a function such that for any $j\in E^c$, $\mathcal{C}_{f(j)}$ is the unique column with $\ell(\mathcal{C}_{f(j)})=\boldsymbol{b}_j$. 
Then, according to Proposition \ref{width_dis}, we have 
\begin{align} \label{inS'}
&\lambda_{\mathcal{S}'}(\boldsymbol{b}_{1} \boldsymbol{b}_{2} \cdots \boldsymbol{b}_{M}) \notag \\
& = (1-\lambda(\mathcal{C}_0)) \left(\sum_{\substack{i_1, \cdots, i_M: \\
\ell(\mathcal{C}_{i_1})\cdots \ell(\mathcal{C}_{i_M})=\boldsymbol{b}_{1} \cdots \boldsymbol{b}_{M}}}\frac{\lambda(\mathcal{C}_{i_1})}{1-\lambda(\mathcal{C}_0)} \cdots \frac{\lambda(\mathcal{C}_{i_M})}{1-\lambda(\mathcal{C}_0)}\right)  \notag \\
&=( 1-\lambda(\mathcal{C}_0))\left(\sum_{\substack{i_j\in \{k-1,k\}: j\in E \\ i_l=f(l): l\in E^c}} \frac{\lambda(\mathcal{C}_{i_1})}{1-\lambda(\mathcal{C}_0)} \cdots \frac{\lambda(\mathcal{C}_{i_M})}{1-\lambda(\mathcal{C}_0)} \right) \notag  \\
&=( 1-\lambda(\mathcal{C}_0)) \left(\frac{\lambda(\mathcal{C}_{k-1})+\lambda(\mathcal{C}_{k})}{1-\lambda(\mathcal{C}_0)}\right)^{\vert E \vert} \left( \prod_{l\in {E}^c} \frac{\lambda(\mathcal{C}_{f(l)})}{1-\lambda(\mathcal{C}_0)}\right).
\end{align}
Similarly, noting that all the columns in $\{\mathcal{C}_1, \cdots, \mathcal{C}_{k-2}, \widetilde{\mathcal{C}}_{k-1}\}$ have distinct labels, we have
\begin{align} \label{intildeS'}
\lambda_{\widetilde{\mathcal{S}}'}(\boldsymbol{b}_{1} \boldsymbol{b}_{2} \cdots \boldsymbol{b}_{M})=( 1-\lambda(\mathcal{C}_0)) \left(\frac{\lambda(\widetilde{\mathcal{C}}_{k-1})}{1-\lambda(\mathcal{C}_0)}\right)^{\vert E \vert} \left( \prod_{l\in {E}^c} \frac{\lambda(\mathcal{C}_{f(l)})}{1-\lambda(\mathcal{C}_0)}\right).
\end{align}
Recalling the fact that $\lambda(\widetilde{\mathcal{C}}_{k-1}) = \lambda(\mathcal{C}_{k-1}) + \lambda (\mathcal{C}_k)$, we conclude from (\ref{inS'}) and (\ref{intildeS'}) that
\begin{align} 
\lambda_{\mathcal{S}'}(\boldsymbol{b}_{1} \boldsymbol{b}_{2} \cdots \boldsymbol{b}_{M})=\lambda_{\widetilde{\mathcal{S}}'}(\boldsymbol{b}_{1} \boldsymbol{b}_{2} \cdots \boldsymbol{b}_{M}) \notag
\end{align}
for Case 2.

Combining the above two cases, we conclude from Definition \ref{label_function} that $H({\mathcal{S}'}) = H({\widetilde{\mathcal{S}}'})$.
\end{proof}

\begin{rem} \label{distribution_is_stable}
It can be immediately checked from the above proof that we indeed have $\boldsymbol{w}(\mathcal{S}') = \boldsymbol{w}(\widetilde{\mathcal{S}}')$. In other words, the width distribution of the resulting gadget will be kept if we interchange the order of merging and independent cutting and stacking  .
\end{rem}

The next lemma reveals that the change of the normalized Shannon entropy will be small as long as the independent cutting and stacking is done on a large fraction of the gadget.

\begin{lem} \label{entropy_rate_of_gadget_1}
Let $\mathcal{S}=\{\mathcal{C}_0, \mathcal{C}_1, \cdots, \mathcal{C}_k\}$ be a gadget labelled over a finite alpbabet such that $\lambda(\mathcal{S})=1, \lambda(\mathcal{C}_0)=\alpha_0$, $h(\mathcal{C}_i)=h$ for any $1\leq i\leq k$ and
$\ell(\mathcal{C}_i)\neq \ell(\mathcal{C}_j)$ for any $1\leq i\neq j \leq k$. We also assume that there exists $1\leq i'\leq k$ such that $\ell(\mathcal{C}_{i'})=\ell(\mathcal{C}_0)$. Let $e$ denote the Euler number and $\mathcal{S}'\triangleq \{\langle \mathcal{C}_0 \rangle_M, \{\mathcal{C}_1, \cdots, \mathcal{C}_k\}^{\langle M \rangle}\}$. 
If
$$
\lambda(\mathcal{C}_0 \cup\{\mathcal{C}_{i'}: \ell(\mathcal{C}_{i'})=\ell(\mathcal{C}_0)\})\leq \frac{1}{e},
$$
then
$$H({\mathcal{S}'}) \geq H({\mathcal{S}})- H_{\textnormal{b}}(\alpha_0),$$
where $H_{\textnormal{b}}(\alpha_0) \triangleq -\alpha_0 \log \alpha_0 - (1-\alpha_0) \log (1-\alpha_0)$ is the binary entropy function.
\end{lem}
\begin{proof}
Without loss of generality, we assume that in the set $\{\mathcal{C}_1, \mathcal{C}_2, \cdots, \mathcal{C}_k\}$, only $\mathcal{C}_k$ has the same label as $\mathcal{C}_0$. Otherwise we can merge columns having the same label as $\mathcal{C}_0$ before performing cutting and stacking, and the resulting normalized Shannon entropy $H(\mathcal{S}')$ will not change due to Lemma~\ref{entropy_rate_of_gadget_2}. 

For any $1\leq i\leq k$, let $\alpha_i \triangleq \lambda(\mathcal{C}_i)$. Then, by definition, we have
\begin{align} \label{H_S=l}
H({\mathcal{S}})=-\frac{1}{h} \left( (\alpha_0+\alpha_k) \log (\alpha_0+\alpha_k) + \sum_{i=1}^{k-1} \alpha_i \log \alpha_i\right).
\end{align}
On the other hand, we see from the definition of cutting and stacking that for any column $\mathcal{D}\in \mathcal{S}'$, there is a sequence $\boldsymbol{b}_1 \boldsymbol{b}_2\cdots \boldsymbol{b}_M\in \{\ell(\mathcal{C}_0), \ell(\mathcal{C}_1) \cdots, \ell(\mathcal{C}_k)\}^M$ such that $\ell(\mathcal{D})=\boldsymbol{b}_1 \boldsymbol{b}_2 \cdots \boldsymbol{b}_M$. Moreover, noting that $\ell(\mathcal{C}_0)=\ell(\mathcal{C}_k)$ and $\ell(\mathcal{C}_i)\neq \ell( \mathcal{C}_j)$ for any $1\leq i\neq j\leq k$, we have 
\begin{align} \label{two_formula}
\begin{cases}
\lambda_{\mathcal{S}'}(\boldsymbol{b}_1 \boldsymbol{b}_2 \cdots \boldsymbol{b}_M)=\alpha_0+(1-\alpha_0)\displaystyle\left(\frac{\alpha_k}{1-\alpha_0}\right)^M \mbox{if $\boldsymbol{b}_j=\ell(\mathcal{C}_0)$ for any $1\leq j\leq M$}  \\
\\
 \lambda_{\mathcal{S}'}(\boldsymbol{b}_1 \boldsymbol{b}_2 \cdots \boldsymbol{b}_M)=(1-\alpha_0)\displaystyle\frac{\alpha_{f(1)}\alpha_{f(2)}\cdots \alpha_{f(k)}}{(1-\alpha_0)^M}\mbox{ if $\boldsymbol{b}_j\neq \ell(\mathcal{C}_0)$ for some $1 \leq j \leq M$},
 \end{cases}
\end{align}
where $\lambda(\mathcal{S}')$ and the function $f$ are defined as in the proof of Lemma \ref{entropy_rate_of_gadget_2}.
Therefore, we can lower bound the normalized Shannon entropy of $\mathcal{S}'$ as
\begin{align}
H({\mathcal{S}'}) &= -\frac{1}{Mh} \Bigg(\lambda_{\mathcal{S}'}(\underbrace{\ell(\mathcal{C}_0) \cdots \ell(\mathcal{C}_0)}_{M})\log(\lambda_{\mathcal{S}'}(\underbrace{\ell(\mathcal{C}_0) \cdots \ell(\mathcal{C}_0)}_{M}))  \notag \\
& \qquad \qquad \quad+ \sum_{\substack{\boldsymbol{b}_1\cdots \boldsymbol{b}_M \in \{\ell(\mathcal{C}_1),\cdots, \ell(\mathcal{C}_k)\}^M: \\ \boldsymbol{b}_j\neq \ell(\mathcal{C}_0)  \mbox{ for some $j$}}} \lambda_{\mathcal{S}'}(\boldsymbol{b}_1 \cdots \boldsymbol{b}_M) \log \lambda_{\mathcal{S}'}(\boldsymbol{b}_1\cdots \boldsymbol{b}_M)\Bigg) \notag
\end{align}
{\small
\begin{align}
&\overset{(g)}{=}-\frac{1}{Mh}\Bigg(\left( \alpha_0+\frac{\alpha_k^M}{(1-\alpha_0)^{M-1}} \right)\log\left( \alpha_0+\frac{\alpha_k^M}{(1-\alpha_0)^{M-1}} \right)  \notag \\
& \qquad \qquad \quad+ \sum_{\substack{1\leq {i_1}, i_2,\cdots, {i_M} \leq k: \\ i_j\neq k  \mbox{ for some $j$}}} \frac{\alpha_{i_1}\alpha_{i_2}\cdots \alpha_{i_M}}{(1-\alpha_0)^{M-1}} \log \frac{\alpha_{i_1}\alpha_{i_2}\cdots \alpha_{i_M}}{(1-\alpha_0)^{M-1}} \Bigg)\notag \\
&\overset{(o)}{=} -\frac{1}{Mh} \Bigg(  \sum_{1\leq i_1, i_2\cdots, i_M\leq k} \frac{\alpha_{i_1} \alpha_{i_2} \cdots \alpha_{i_M}}{(1-\alpha_0)^{M-1}} \log \frac{\alpha_{i_1} \alpha_{i_2} \cdots \alpha_{i_M}}{(1-\alpha_0)^{M-1}}\Bigg) \notag \\
&= -\frac{1}{Mh} \Bigg(  \sum_{1\leq i_1, \cdots, i_M\leq k} \frac{\alpha_{i_1} \cdots \alpha_{i_M}}{(1-\alpha_0)^{M-1}} \log \frac{1}{(1-\alpha_0)^{M-1}} + \sum_{1\leq i_1, \cdots, i_M\leq k} \frac{\alpha_{i_1} \cdots \alpha_{i_M}}{(1-\alpha_0)^{M-1}} \log \alpha_{i_1} \alpha_{i_2} \cdots \alpha_{i_M}\Bigg) \notag \\
&\overset{(p)}{=} -\frac{1}{Mh} \Bigg(  - (M-1) (1-\alpha_0) \log (1-\alpha_0) + M \sum_{i=1}^k \alpha_i \log \alpha_i \Bigg) \notag \\
&=-\frac{1}{Mh} \Bigg( M (\alpha_0+\alpha_k) \log (\alpha_0+ \alpha_k) +M \sum_{i=1}^{k-1} \alpha_i \log \alpha_i -M (\alpha_0+\alpha_k) \log (\alpha_0+ \alpha_k) \notag \\
& \qquad \qquad + M \alpha_k \log \alpha_k -(M-1) (1-\alpha_0) \log (1-\alpha_0) \Bigg) \notag \\
&\overset{(q)}{=} H({\mathcal{S}}) - \frac{M(- (\alpha_0+\alpha_k) \log (\alpha_0+ \alpha_k)
+  \alpha_k \log \alpha_k) -(M-1) (1-\alpha_0) \log (1-\alpha_0)}{Mh} \notag \\
&\overset{(r)}{\geq} H({\mathcal{S}}) - \frac{M(- \alpha_0 \log \alpha_0) -(M-1) (1-\alpha_0) \log (1-\alpha_0)}{Mh} \notag \\
&\geq H(\mathcal{S}) - H_{\textnormal{b}}(\alpha_0) \notag,
\end{align}
}
\!\!where $(g)$ follows from (\ref{two_formula}), $(o)$ follows from the fact that $\alpha_0+\alpha_k^M/(1-\alpha_0)^{M-1}<\lambda(\mathcal{C}_0)+\lambda(\mathcal{C}_k)\leq 1/e$ and $-x\log x$ is an increasing function for $0<x<1/e$, 
$(p)$ follows from the fact that $\sum_{{j}=1}^k \alpha_{j} =1-\alpha_0$, $(q)$ follows from (\ref{H_S=l}), and $(r)$ follows from the fact that $-(x+y) \log (x+y) \leq -x \log x - y \log y$ for any $x, y>0$. Hence the lemma is proved.
\end{proof}

The above two lemmas immediately imply the following corollary, which gives a useful relationship between the normalized Shannon entropy of $\mathcal{G}(m)$ and $\mathcal{G}(m+1).$

\begin{co} \label{S(m)_and_S(m+1)}
Let $\{\mathcal{G}(m)\}_{m=1}^\infty$ be the sequence of gadgets constructed in Section~\ref{constru_of_counter}. Then,
$$
H({\mathcal{G}(m+1)}) \geq H({\mathcal{G}(m)})-H_{\textnormal{b}}(\beta_{m+1}).
$$
\end{co}
\begin{proof}
For any $m\geq 1$, let $\mathcal{G}(m)$, $\mathcal{L}(m),$ $\mathcal{R}(m), \mathcal{L}(m,1)$ and $\mathcal{L}(m,2)$ be defined as in Section \ref{constru_of_counter}. We further define 
\begin{align*}
\mathcal{G}^{(1)}(m)&\triangleq \{\mathcal{L}(m,1), \{\mathcal{L}(m,2), \mathcal{R}(m)\}\}, \\
\mathcal{G}^{(2)}(m)&\triangleq \{\mathcal{L}(m,1), \mathcal{R}'(m)\}, \\
\mathcal{G}^{(3)}(m)&\triangleq \{\langle\mathcal{L}(m,1)\rangle_{l_{m+1}/m_m}, \mathcal{R}'(m)^{\langle l_{m+1}/l_m\rangle}\},
\end{align*}
where $\mathcal{R}'$ is the merging of 
$\{\mathcal{L}(m,2) , \mathcal{R}(m)\}$. 
By Definition \ref{label_function}, we see that $H({\mathcal{S}(m)}) \triangleq H(\mathcal{G}(m))= H({\mathcal{G}^{(1)}(m)})=H({\mathcal{G}^{(2)}(m)})$; on the other hand, 
we infer from Lemma \ref{entropy_rate_of_gadget_2} that $H({\mathcal{G}^{(3)}(m)})=H({\mathcal{G}(m+1)})$. Finally, noting that Lemma \ref{entropy_rate_of_gadget_1} implies
$H({\mathcal{G}^{(3)}(m)})\geq H({\mathcal{G}^{(2)}(m)})-H_{\textnormal{b}}(\lambda(\mathcal{L}{(m,1)})),$ we conclude that
\begin{align*}
H({\mathcal{G}(m+1)})= H({\mathcal{G}^{(3)}(m)}) \geq H({\mathcal{G}^{(2)}(m)})-H_{\textnormal{b}}(\lambda(\mathcal{L}{(m,1)}))=H({\mathcal{G}(m)})- H_{\textnormal{b}}(\beta_{m+1}),
\end{align*}
as desired.
\end{proof}
Our next goal is to develope the relationship between the normalized Shannon entropy of a sequence of gadgets and that of its corresponding final process. We begin with some definitions.

\begin{de} \label{conca_def}
Let $X_1^N=(X_1, X_2, \cdots, X_N)$ be a random vector taking values in $\mathcal{A}^N$ and let $\mu$ be the distribution of $X_1^N$ and therefore $\mu$ is a measure on $\mathcal{A}^N$. 
Moreover, let $\psi: (\mathcal{A}^N)^\infty \rightarrow \mathcal{A}^\infty$ be the mapping given by $\psi(y_1^{N}, y_{N+1}^{2N}, y_{2N+1}^{3N} \cdots) = y_1y_2y_3\cdots$ for any choice of blocks $y_1^{N}, y_{N+1}^{2N}, y_{2N+1}^{3N} \cdots$ in $\mathcal{A}^N$. Then the {\em concatenated-block process defined by $X_1^N$} is a process $\{\widetilde{Y}_n\}_{n=1}^\infty$ such that for any $k$ and any $a_1^k\in \mathcal{A}^k$, $\widetilde{Y}_1^k=a_1^k$ with distribution $\nu(a_1^k)$, where 
\begin{align} \label{measure_relation}
\nu(a_1^k)\triangleq \frac{1}{N} \sum_{i=1}^{N-1} \mu^*(\psi^{-1} (\varphi^{-i} \{x_1^\infty: x_1^k=a_1^k\})),
\end{align}
$\mu^*$ is the infinite-product measure of $\mu$, $\varphi$ is the shift operator on the sequence space $\mathcal{A}^\infty$ and for any measurable set $B$, $\varphi^{-1} B \triangleq \{x: \varphi x \in B\}$. 
\end{de}

\begin{rem} \label{block-inde}
Let $X_1^N$ be defined as in Definition \ref{conca_def} and let $\{Y_n\}_{n=1}^\infty$ be a process such that the blocks $\{(Y_{iN+1}, Y_{iN+2}, \cdots, Y_{(i+1)N}):i\geq 0\}$ are independent and for each $i$ and each $a_1^N\in \mathcal{A}^N$, 
$Y_{iN+1}^{(i+1)N}=a_1^N$ with probability $\mu(a_1^N)$. Then, the concatenated-block process defined by $X_1^N$ can be obtained from $\{Y_n\}_{n=1}^\infty$ by ``randomizing the start". The process $\{Y_n\}_{n=1}^\infty$ here is sometimes called {\em the block-independent process of $X_1^N$} (see, for example, \cite{sh96}).
\end{rem}

Concatenated-block processes are intimately related to gadgets in the following manner. Let 
$\mathcal{S}=\{\mathcal{C}_1, \mathcal{C}_2, \cdots, \mathcal{C}_k\}$ be a gadget labelled over the alphabet $\mathcal{A}$ such that $h(\mathcal{C}_i)=h$ for any $1\leq i\leq k$ and $\lambda(\mathcal{S})=1$.
Define a special sequence of gadgets $\{\mathcal{S}(m)\}_{m=1}^\infty$ by letting $\mathcal{S}(1) \triangleq \mathcal{S}$ and $\mathcal{S}(m+1)\triangleq \mathcal{S}(m)^{\langle 2 \rangle}$ 
for any $m\geq 1$.
As we mentioned before, for any $m\geq 1$, the label of any column in $\mathcal{S}(m)$ is a concatenation of the label of $2^m$ shorter blocks in $\mathcal{S}$; and moreover, we infer from Proposition \ref{width_dis} that the measure distribution of $\mathcal{S}(m+1)$ is the
$2^m$-fold Kronecher product of the measure distribution $\boldsymbol{\lambda}(\mathcal{S})$. Hence, 
letting $X_1^h$ be the random vector such that $X_1^h=\ell(\mathcal{C}_i)$ with probability $\lambda(\mathcal{C}_i)$ for any $1\leq i\leq k$, 
we see that the $(T, \mathcal{P}_{\mathcal{S}(1)})$ process (denoted by $\nu_{\mathcal{S}}$) as elaborated in Theorem \ref{pro_of_gad} is the concatenated-block process defined by $X_1^h$.

\begin{rem}
It is worth noting that the process $\nu_{\mathcal{S}}$ is completely determined by the gadget $\mathcal{S}$. To highlight this dependence, in the remainder of this section, we  call $\nu_{\mathcal{S}}$ {\em the concatenated-block process given by $\mathcal{S}$}.
\end{rem}

Our next result says that the Shannon entropy rate of the concatenated-block process given by the gadget ${\mathcal{S}}$ is equal to the normalized Shannon entropy of $\mathcal{S}.$
\begin{lem} \label{conca_pro_entropy}
Let $\mathcal{S}=\{\mathcal{C}_1,\mathcal{C}_2,\cdots,\mathcal{C}_k\}$ be a labelled gadget such that $\lambda(\mathcal{S})=1$ and $h(\mathcal{C}_i)=h$ for any $1\leq i\leq k$. We further let $\nu_{\mathcal{S}}$ be the concatenated-block process given by $\mathcal{S}$. Then $$H(\nu_{\mathcal{S}}) = H(\mathcal{S}),$$
where $H(\nu_{\mathcal{S}})$ is the Shannon entropy rate of
$\nu_{\mathcal{S}}$ and $H(\mathcal{S})$ is the normalized Shannon entropy of $\mathcal{S}$.
\end{lem}

\begin{proof}
Without loss of generality, we assume that columns in $\mathcal{S}$ have distinct labels (otherwise, before performing the independent cutting and stacking, we can merge columns with the same label and Remark \ref{distribution_is_stable} implies that this keeps the width distribution of the resulting gadget). Let $\{Y_n\}$ be the block-independent process of $\mathcal{S}$ as defined in Remark \ref{block-inde}. For any $n$, write $n=k l + r$ where $k, r$ are nonnegative integers and $1\leq r \leq l-1$. Then by definition,
\begin{align}
H(Y) &= \lim_{n\rightarrow \infty}\frac{1}{n} H(Y_1, Y_2, \cdots Y_{n}) \notag \\
&= \lim_{k\rightarrow \infty}  \frac{k H(Y_1, Y_2, \cdots, Y_l)+H(Y_{k l+1}, \cdots Y_{kl +r}) }{k l + r}  \notag  \\
&= \frac{H(Y_1,Y_2, \cdots Y_l)}{l} \notag \\
&= H(\mathcal{S}),
\end{align}
where $H(Y)$ is the Shannon entropy rate of $\{Y_n\}$.
On the other hand, recalling from Remark \ref{block-inde} that $\nu_S$ is obtained from $\{Y_n\}$ by randomizing the start, 
we derive that $\nu_{\mathcal{S}}$ and $\{Y_n\}$ have the same Shannon entropy rate \cite{sh96}. Hence, we have $H(\nu_{\mathcal{S}})= H(Y) = H({\mathcal{S}}),$ proving the lemma.
\end{proof}

Now suppose the sequence of gadgets $\{\mathcal{G}(m)\}_{m=1}^\infty$ (labelled over the alphabet $\{0,1\}$) constructed in Section \ref{constru_of_counter} satisfies the conditions in Theorem \ref{cut_and_erg}. On the one hand, each fixed gadget $\mathcal{G}(m)$ from this sequence gives a binary block-independent process $\nu_{\mathcal{G}(m)}$; on the other hand, Theorem \ref{cut_and_erg} implies that there is a binary ergodic final process $\mu_\infty^{\mathcal{G}}$ corresponding to this sequence. Indeed, we have the following lemma.

\begin{lem} \label{weak_conv}
$\nu_{\mathcal{G}(m)}$ converges weakly to $\mu_\infty^{\mathcal{G}}$ as $m$ goes to infinity.
\end{lem}

\begin{proof}
Since both $\nu_{\mathcal{G}(m)}$ and $\mu_f^{\mathcal{G}}$ are binary stationary processes, it suffices for us to show that for any $k$ and any $a_1^k\in \{0,1\}$,
\begin{align*}
\lim_{m\rightarrow \infty} \nu_{\mathcal{G}(m)} (a_1^k) = \mu_\infty^{\mathcal{G}}(a_1^k).
\end{align*}
To this end, for any $\epl>0$ and any fixed $k\in \mathbb{N}^+$, we first choose $m\in \mathbb{N}+$ such that $k/l_m<\epl$ where $l_m$ is the height of columns in $\mathcal{G}(m)$. Then, we let $\mu^{(l_m)}$ be the measure distribution (which is defined on $\{0,1\}^{l_m}$) of 
$\mathcal{G}(m)$ and let $\mu_{\mathcal{G}(m)}^*$ be the distribution of the block-independent process defined by $\mathcal{G}(m)$, which, by definition, is the infinite-product measure of $\mu^{(l_m)}$ on $\{0,1\}^\infty$. We further let $\psi^{(m)}: (\{0,1\}^{l_m})^\infty \rightarrow  \{0,1\}^\infty$ be the mapping given by $\psi(w(1),w(2),\cdots)=w(1)w(2)\cdots$ for any choice of $w(1),w(2),\cdots$ in $\mathcal{A}^{l_m}$ and
$\varphi$ be the left shift. 
Then, we deduce from (\ref{measure_relation}) that for any $a_1^k\in \{0,1\}^k$,
{\small
\begin{align}\label{measure_1}
\nu_{\mathcal{G}(m)}(a_1^k)&= \frac{1}{l_m} \sum_{i=0}^{l_m-1} \mu_{\mathcal{G}(m)}^* (\psi^{-1}(\varphi^{-i}\{x_1^\infty: x_1^k=a_1^k\})) \notag \\
&=\frac{1}{l_m} \left( \sum_{i=0}^{l_m-k} \mu_{\mathcal{G}(m)}^* (\psi^{-1} (\varphi^{-i} \{x_1^\infty: x_1^k=a_1^k\})) + \sum_{i=l_m-k+1}^{l_m-1}  \mu_{\mathcal{G}(m)}^* (\psi^{-1} (\varphi^{-i} \{x_1^\infty: x_1^k=a_1^k)) \right) \notag \\
&\overset{(s)}{=} \frac{1}{l_m} \left( \sum_{i=0}^{l_m-k} \mu^{(l_m)} (\{x_1^{l_m}: x_{i+1}^{i+k} = a_1^k\})  + \sum_{i=l_m-k+1}^{l_m-1}  (\mu^{(l_m)} \times \mu^{(l_m)})(\{x_1^{2 l_m}: x_{i+1}^{i+k} = a_1^k\}) \right) \notag \\
&\overset{(t)}{=} \frac{1}{l_m}  \sum_{i=0}^{l_m-k} \sum_{\mathcal{C}\in \mathcal{G}(m)} \mathbf{1}_{\{\ell(\mathcal{C})_{i+1}^{i+k}=a_1^k\}}(\mathcal{C})  \lambda(\mathcal{C}) + \frac{1}{l_m} \sum_{i=l_m-k+1}^{l_m-1} \sum_{\mathcal{D} \in \mathcal{G}(m)^{\langle 2 \rangle}}\mathbf{1}_{\{\ell(\mathcal{D})_{i+1}^{i+k}=a_1^k\} } (\mathcal{D})\lambda(D),
\end{align}
}
\!\!where $\ell(\mathcal{C})_{i}^j$ is the subsequence consists of entries from the $i$-th position to the $j$-th position of the sequence $\ell(\mathcal{C})$ for any $1\leq i\leq j\leq l_m$, $(s)$ follows from the fact that $\mu_{\mathcal{G}(m)}^*$ is the product measure of $\mu^{(l_m)}$ and $(t)$ follows from the definition of the distribution of labeled sequences in a gadget.
Observing that
$$\frac{1}{l_m} \sum_{i=l_m-k+1}^{l_m-1} \sum_{\mathcal{D} \in \mathcal{G}(m)^2}\mathbf{1}_{\{\ell(\mathcal{D})_{i+1}^{i+k}=a_1^k\} } \lambda(\mathcal{D})\leq \frac{k}{l_m} \sum_{\mathcal{D}\in \mathcal{G}(m)^2} \lambda(D) =\frac{k}{l_m} <\epl,$$
we derive that
$$
\frac{1}{l_m} \sum_{i=l_m-k+1}^{l_m-1} \sum_{\mathcal{D} \in \mathcal{G}(m)^2}\mathbf{1}_{\{\ell(\mathcal{D})_{i+1}^{i+k}=a_1^k\} } \lambda(\mathcal{D}) \rightarrow 0 \quad \mbox{as $m\rightarrow \infty$ }
$$
due to the arbitrariness of $\epl$. It then follows from (\ref{measure_1}) that
\begin{align}
\lim_{m\rightarrow \infty} \nu_{\mathcal{G}(m)} (a_1^k) &=\lim_{m\rightarrow \infty} \frac{1}{l_m}  \sum_{i=0}^{l_m-k} \sum_{\mathcal{C}\in \mathcal{G}(m)} \mathbf{1}_{\{\ell(\mathcal{C})_{i+1}^{i+k}=a_1^k\}}(\mathcal{C})  \lambda(\mathcal{C})  \notag \\
&=\lim_{m\rightarrow \infty} \frac{l_m-k+1}{l_m} \sum_{\mathcal{C}\in \mathcal{G}(m)} \frac{\sum_{i=0}^{l_m-k} \mathbf{1}_{\{\ell(\mathcal{C})_{i+1}^{i+k} =a_1^k\}}(\mathcal{C})}{l_m-k+1} \lambda(\mathcal{C}) \notag \\
&=\lim_{m\rightarrow \infty} \frac{l_m-k+1}{l_m} \sum_{\mathcal{C}\in \mathcal{G}(m)} p_k(a_1^k\vert \mathcal{C}) \lambda(\mathcal{C}) \notag \\
&= \mu_\infty^\mathcal{G}(a_1^k), \notag
\end{align}
where the last equality follows from Theorem \ref{joint_dis}, proving the lemma.
\end{proof}

Based on all the lemmas and results given above, we are prepared to complete the proof of Proposition \ref{property_of_mu}.
\medskip

\noindent\textbf{{Proof of Proposition \ref{property_of_mu}:} }
The proofs of $(B)$ and $(C)$ for any $\{\alpha_m\}_{m=1}^\infty$ are similar to those given in Section III.1.c of \cite{sh96}. We sketch the proof here for completeness. First, noting that for any $m$ and sample point $\omega$ in the support of $\mathcal{L}(m)$ yet not in the top $m$ levels, $T^j \omega\in \mathcal{P}_1$ for any $0\leq j \leq m-1$, where $\mathcal{P}_1$ is defined as in Definition \ref{par_map_def}.
Hence, we derive by definition that
$$\mu_\infty^\mathcal{G}(\{x_1^\infty: x_1^m= \underbrace{1 \cdots 1}_{m}\}) \geq \left( 1-\frac{m}{l_m}\right) \beta_m\geq \alpha_m,$$
where the last inequality follows from Step 4 of the construction in Section \ref{constru_of_counter}. This proves $(B)$.
To prove $(C)$, recall from Condition $(c)$ in Step 4 of the construction in Section \ref{constru_of_counter} that $\mathcal{R}(m+1)$ and $\{\mathcal{L}(m,2), \mathcal{R}(m)\}$ are $\epl_m$-independent, where $\epl_m$ goes to $0$ as $m$ goes to infinity. This, by direct calculation, implies that $\mathcal{G}(m)$ and $\mathcal{G}(m+1)$ are $(\epl_m+\beta_m+\beta_{m+1})$-independent. Since $\epl_m+\beta_m+\beta_{m+1}$ goes to $0$ as $m$ goes to infinity, Theorem \ref{cut_and_erg} implies that $\mu_\infty^{\mathcal{G}}$ is an ergodic process, proving $(C)$.

The proofs of $(A)$ and $(D)$ are more involved. First we note from the construction of $\mu_\infty^{\mathcal{G}}$ that it is the $(T, \mathcal{P}_{\mathcal{G}(1)})$-process given by the sequence $\{\mathcal{G}(m)\}_{m=1}^\infty$. For each $m$, letting $\nu_{\mathcal{G}(m)}$ be the concatenated-block process given by $\mathcal{G}(m)$, we deduce from
Lemma \ref{weak_conv} that
$\nu_{\mathcal{G}(m)}$ converges weakly to $\mu_\infty^{\mathcal{G}}$, as $m$ goes to infinity. Recalling that the entropy rate is weakly upper semi-continuous with respect to the weak topology on the space of stationary probability measures (see Theorem I.9.1 of \cite{sh96}), we have
\begin{align*}
H(\mu_\infty^{\mathcal{G}}) \geq \limsup_{n\rightarrow \infty} H(\nu_{\mathcal{G}(m)})=\limsup_{n\rightarrow \infty} H({\mathcal{G}(m)}),
\end{align*}
where the last equality follows from Lemma \ref{conca_pro_entropy}. Hence, in order to prove $H(\mu_\infty^{\mathcal{G}})>1/2$, it suffices to show $H_{\mathcal{G}(m)}>1/2+\delta$ for some $\delta>0$.

To this end,
we first note from the definition of $\mathcal{G}(1)$ 
that the normalized Shannon entropy of $\mathcal{G}(1)$ can be computed as
\begin{align*}
H({\mathcal{G}(1)}) &= -\frac{1}{l_1} \left( \frac{2}{2^{2 l_1/3}} \log \frac{2}{2^{2l_1/3}} + (2^{2l_1/3}-2)\frac{1}{2^{2l_1/3}} \log \frac{1}{2^{2 l_1 /3}}\right) \\
&=-\frac{1}{l_1} \left( \frac{2}{2^{2l_1 /3}} \left(\frac{2 l_1}{3} -1 \right) +\frac{2 l_1}{3} -\frac{2}{2^{2l_1/3}} \frac{2 l_1}{3}\right) \\
&= \frac{2}{3}- \frac{1}{ l_1 2^{2l_1/3 -1}}.
\end{align*}
Then, recalling from Corollary \ref{S(m)_and_S(m+1)} that
 $H({\mathcal{G}(m+1)}) \geq H({\mathcal{G}(m)})-H_{\textnormal{b}}(\beta_{m+1})$, we have
\begin{align} \label{H_{sm}}
H({\mathcal{G}(m)}) \geq H({\mathcal{G}(1)}) - \sum_{m=2}^\infty H_{\textnormal{b}}(\beta_m) \geq \frac{2}{3}-\frac{1}{l_1 2^{2l_1/3}-1}-\sum_{m=2}^\infty H_{\textnormal{b}}(\beta_m)
\end{align}
for any $m\geq 1$ by an induction argument.
We also observe that
\begin{align} 
\sum_{m=1}^\infty H_{\textnormal{b}}(1/m^2) =\sum_{m=1}^\infty \left(-\frac{1}{m^2} \log \frac{1}{m^2} -\frac{m^2-1}{m^2} \log \frac{m^2-1}{m^2}\right) \notag 
\end{align}
\begin{align} \label{sum_of_H(1/m^2)}
&= \sum_{m=1}^\infty \left( \frac{2}{m^2} \log m + \frac{m^2-1}{m^2} \log \left(1+ \frac{1}{m^2-1} \right) \right) \notag \\
&= \sum_{m=1}^\infty \left( \frac{2}{m^2} \log m + \frac{m^2-1}{m^2} \log \left(1+ \frac{1}{m^2-1} \right) \right) \notag \\
& \overset{(k)}{\leq} \sum_{m=1}^\infty \left( \frac{4}{m^{3/2}} + \frac{1}{m^2-1} \right) \notag \\
&<\infty,
\end{align}
where $(k)$ follows from the fact that $\log x =2 \log \sqrt{x} \leq 2 \sqrt{x}$ for any $x>0$. Now, (\ref{sum_of_H(1/m^2)}), together with the fact that $\displaystyle\frac{1}{l_1 2^{2 l_1/3} -1}\rightarrow 0$ as $l_1 \rightarrow \infty$, implies that there exist $l_1$ and $N$ such that
\begin{align} \label{less_than_1/6}
\frac{1}{l_1 2^{2l_1/3}-1}+\sum_{m=2}^\infty H_{\textnormal{b}}(1/(m+N)^2) < \frac{1}{6}.
\end{align}
Finally, for any $m\geq 1$, choose $\alpha_m\triangleq 1/(m+N)^3$, $\beta_m\triangleq 1/(m+N)^2$ and let $l_1$ be chosen as above. Then $(A)$ obviously holds and we conclude from (\ref{H_{sm}}) that $H_{\mathcal{G}(m)} >1/2$ for any $m\geq 1$, proving $(D)$.
\qed

\subsection{Gap between $\lim_{\alpha\rightarrow 1^+} H_{\alpha}(\mu_\infty^\mathcal{G})$ and $H(\mu_\infty^\mathcal{G})$} \label{Non-existence_of_rate}
Using the results from the previous section, we show in this subsection that for $l_1$, $\{\alpha_m\}_{m=1}^\infty$ and $\{\beta_m\}_{m=1}^\infty$ chosen as in the proof of Proposition \ref{property_of_mu}, the R\'enyi entropy rate of the final process $\mu_\infty^{\mathcal{G}}$ does not converge to its Shannon entropy rate.
More explicitly, we have the following theorem.

\begin{thm} \label{dis_of_mu}
Let $\mu_\infty^\mathcal{G}$ be the final process of the gadgets $\{\mathcal{G}(m)\}_{m=1}^\infty$ constructed in Section \ref{constru_of_counter}. If $l_1$, $\{\alpha_m\}_{m=1}^\infty$ and $\{\beta_m\}_{m=1}^\infty$ are chosen such that Proposition \ref{property_of_mu} holds,
then $H_{\alpha}(\mu_\infty^\mathcal{G})$ exists for $\alpha\in [1, \infty)$, 
$H_{\alpha}(\mu_\infty^\mathcal{G})=0$ for $\alpha>1$, and $\lim_{\alpha\rightarrow 1^+}H_{\alpha} (\mu_\infty^\mathcal{G})< H(\mu_\infty^\mathcal{G})$.
\end{thm}

\begin{proof}
Since the case for $\alpha=1$ corresponds to the Shannon entropy rate which is well-defined for all stationary processes, we focus on the case $\alpha>1$ in what follows.

Recalling from Proposition \ref{property_of_mu} that 
$
\mu_\infty^\mathcal{G}(\{x_1^\infty: x_1^m=\underbrace{1\cdots 1}_{m}\})\geq \alpha_m,
$
we have
\begin{align*}
\limsup_{m\rightarrow \infty}\frac{1}{(1-\alpha) m} \log \sum_{x_1^m} (\mu_\infty^\mathcal{G}(x_1^m))^{\alpha} &\leq \limsup_{m\rightarrow \infty}\frac{1}{(1-\alpha) m} \log (\mu_\infty^\mathcal{G}(\{x_1^\infty: x_1^m=\underbrace{1 \cdots 1}_{m}\}))^{\alpha} \\
&\leq \limsup_{m\rightarrow \infty}\frac{\alpha}{(1-\alpha) m} \log \alpha_m \\
&=0.
\end{align*}
Observing that 
$$\liminf_{m\rightarrow \infty}\frac{1}{(1-\alpha) m} \log \sum_{x_1^m} (\mu_\infty^\mathcal{G}(x_1^m))^{\alpha}\geq 0,$$
we conclude $H_{\alpha}(\mu_\infty^\mathcal{G})=0$ for $\alpha>1$ and the theorem follows since $H(\mu_\infty^\mathcal{G})>1/2$ from Proposition \ref{property_of_mu}.
\end{proof}

{\bf Acknowledgement.} The last author would like to thank Venkat Anantharam for insightful discussions and pointing out relevant references.

\end{document}